\documentclass{tac}
\usepackage{amssymb}
\usepackage{amsfonts}
\usepackage[matrix]{xy}
\usepackage{lscape}
\input diagxy

\usepackage{amsmath}
\usepackage{amsfonts}
\usepackage{amssymb}
\usepackage{epsfig}
\usepackage[mathscr]{euscript}
\usepackage{lscape}
\usepackage{qcircuit}
\usepackage{circuitikz}
\usepackage{bbm}
\usepackage{old-arrows}
\usepackage[compat=1.0.0]{tikz-feynman}
\usepackage[shortlabels]{enumitem}
\usetikzlibrary{shapes.gates.logic.US, calc}

\usepackage{everypage}
\usepackage{graphicx}
\usepackage{lipsum}
\input xypic
\xyoption{all}
\xyoption{2cell}
\UseAllTwocells
\LaTeXdiagrams

\usepackage{answers}
\Newassociation{sola}{Solution}{ans}

\usepackage{pdflscape}

\newtheorem{defi}{Definition}
\newtheorem{theo}{Theorem}

\newtheorem{tech}{Technical Point}
\newtheorem{examp}{Example}

\def\sq{\sqrt}
\def\sq2{\sqrt{2}}
\def\sq12{\sq{12}}
\def\dsq2{\f{1}{\sqrt{2}}}
\def\be{\begin{equation}}
\def\ee{\end{equation}}
\def\lra{\longrightarrow}
\def\la{\langle}
\def\ra{\rangle}

\def\q={=}

\def\n{{\mathcal N}}
\def\r{{\mathcal R}}

\def\q{{\mathcal Q}}

\newcommand{\f}[2]{\frac{\displaystyle #1}{\displaystyle #2}}

\catcode`\@=11\relax
\newdimen\p@renwd
\font\tenex=cmex10 \setbox0=\hbox{\tenex B} \p@renwd=\wd0
\def\bbordermatrix#1{\begingroup \m@th
\setbox\z@\vbox{\def\\{\crcr\noalign{\kern2\p@\global\let\cr\endline}}%
    \ialign{$##$\hfil\kern2\p@\kern\p@renwd&\thinspace\hfil$##$\hfil
      &&\quad\hfil$##$\hfil\crcr
      \omit\strut\hfil\crcr\noalign{\kern-\baselineskip}%
      #1\crcr\omit\strut\cr}}%
  \setbox\tw@\vbox{\unvcopy\z@\global\setbox\@ne\lastbox}%
  \setbox\tw@\hbox{\unhbox\@ne\unskip\global\setbox\@ne\lastbox}%
  \setbox\tw@\hbox{$\kern\wd\@ne\kern-\p@renwd\left[\kern-\wd\@ne
    \global\setbox\@ne\vbox{\box\@ne\kern2\p@}%
    \vcenter{\kern-\ht\@ne\unvbox\z@\kern-\baselineskip}\,\right]$}%
  \null\;\vbox{\kern\ht\@ne\box\tw@}\endgroup}
\catcode`\@=12\relax

\usepackage{makeidx}
\makeindex

\def\be{\begin{equation}}
\def\ee{\end{equation}}
\def\bt{\begin{tabbing}}
\def\et{\end{tabbing}}
\def\lra{\longrightarrow}

\def\eu{\exists!}

\def\TotCompFunc{\mathbbm{TotCompFunc}}
\def\CompFunc{\mathbbm{CompFunc}}
\def\Func{\mathbbm{Func}}
\def\TotCompString{\mathbbm{TotCompString}}
\def\CompString{\mathbbm{CompString}}
\def\TotTuring{\mathbbm{TotTuring}}
\def\Turing{\mathbbm{Turing}}
\def\NTotTuring{\mathbbm{NTotTuring}}
\def\BoolComp{\mathbbm{CompBool}}
\def\TotBoolComp{\mathbbm{TotCompBool}}
\def\Circuit{\mathbbm{Circuit}}

\def\CircuitFam{\mathbbm{CircuitFam}}
\def\TotCircuitFam{\mathbbm{TotCircuitFam}}

\def\Logic{\mathbbm{Logic}}
\def\PRCompN{\mathbbm{PRComp}\n}
\def\RegMachine{\mathbbm{RegMachine}}
\def\TotRegMachine{\mathbbm{TotRegMachine}}
\def\CompN{\mathbbm{Comp}\n}
\def\TotCompN{\mathbbm{TotComp}\n}
\def\Program{\mathbbm{Program}}
\def\InpTM{\mathbbm{InpTM}}
\def\Const{\mathbbm{Const}}
\def\Poly{\mathbbm{Poly}}
\def\Log{\mathbbm{Log}}
\def\Exp{\mathbbm{Exp}}

\def\NP{\mathbbm{NP}}
\def\NPComplete{\mathbbm{NPComplete}}
\def\PSPACE{\mathbbm{PSPACE}}
\def\NPSPACE{\mathbbm{NPSPACE}}
\def\Algorithm{\mathbbm{Algorithm}}
\def\Set{\mathbbm{Set}}

\def\PT{\mathbbm{P}}

\def\StoMat{\mathbbm{StoMat}}
\def\StoMat1{\mathbbm{StoMat1}}

\def\one{\mathbf{1}}

\newcounter{examnum}[section]
\newcounter{remarnum}[section]
\setenumerate{noitemsep}
\setitemize{noitemsep}

\newenvironment{itemize*}%
  {\begin{itemize}%
    \setlength{\itemsep}{0pt}%
    \setlength{\parskip}{0pt}}%
  {\end{itemize}}

\begin{document}
\title{Theoretical Computer Science  for\\ the Working Category Theorist }
\author{Noson S. Yanofsky}\address{Department of Computer and Information Science,
Brooklyn College, The City University of New York, Brooklyn, N.Y. 11210.
And the Computer Science
Department of the Graduate Center, CUNY, New York, N.Y. 10016. } 
\eaddress{noson@sci.brooklyn.cuny.edu.} 

\maketitle
\begin{abstract}
\noindent Theoretical computer science discusses foundational issues about computations. It asks and answers questions such as ``What is a computation?'', ``What is computable?'', ``What is efficiently computable?'',``What is information?'',  ``What is random?'', ``What is an algorithm?'', etc.  We will present many of the major themes and theorems with the basic language of category theory. Surprisingly, many interesting theorems and concepts of theoretical computer science are easy consequences of functoriality and composition when you look at the right categories and functors connecting them.  
\end{abstract}

\section{Introduction}
From a broadly philosophical perspective, theoretical computer science is the study of the relationship between the syntax and the semantics of functions. By the syntax of a function we mean a description of the function such as a program that implements the function, a computer that ``runs''  the function, a logical formula that characterizes  the function, a circuit that executes the function, etc. By the semantics of a function we mean the rule that assigns to every input an output. There are many aspects to the relationship between the syntax and the semantics of a function. Computability theory asks what functions are defined by syntax, and --- more interestingly --- what functions are not defined by syntax. Complexity theory asks how can we classify and characterize functions by examining their semantics. Kolmogorov complexity deals with the syntax of functions that  only output single strings. Algorithms exist on the thin line between syntax and semantics of computable functions. They are at the core of computer science.     

In a categorical setting, the relationship between the syntax and semantics of functions is described by a functor from a category of syntax to a category of semantics. The functor takes a description of a function to the function it describes. Computability theory then asks what is in the image of this functor and --- more interestingly --- what is not in the image of the functor. Complexity theory tries to classify and characterize what is in the image of the functor by examining the preimage of the functor. Kolmogorov complexity theory does this for functions that output strings. We will classify some functions as compressible and some as random. Since algorithms are between syntax and semantics, the functor from syntax to semantics factors as 
\be Syntax \lra Algorithms \lra Semantics. \ee

This mini-course will flesh-out these ideas. The categories of syntax and semantics are given in in Figure \ref{fig:modelsofcomputation}. The central horizontal line is the core of the semantics of functions. This central line is surrounded by other, equivalent categories of semantics. Three different types of syntax are given on the outside of the spokes. It is essentially irrelevant which syntax is studied. We choose to concentrate on the top spoke of the diagram. 

Major parts of theoretical computer science will be omitted. For example, we will not cover any formal language theory, semantics, and analysis of algorithms. We do not have the space to cover all the subtopics of theoretical computer science and only deal with the central issues in theoretical computer science. 

I am thankful to Gershom Bazerman, Deric Kwok, Florian Lengyel, Armando Matos, Rohit Parikh, and all the members of The New York City Category Theory Seminar for helpful discussions and editing.  

\section{Models of Computation}
The first question we deal with is ``What is a computation?'' We all have a pretty good intuition that a computation is a process that a computer performs. Computer scientists who study this area have given other, more formal, definitions of a computation. They have described different models where computations occur. These models are virtual computers that are exact and have a few simple rules. 

We have united all the different models that we will deal with in Figure \ref{fig:modelsofcomputation} which we call ``The Big Picture.'' This diagram has a center and three spokes coming out of it. This first part of this mini-course will go through the entire diagram. The rest of the mini-course will concentrate on the top spoke. 

Let us give some orientation around The Big Picture so that it is less intimidating. All the categories are symmetric monoidal categories. All the functors are symmetric monoidal functors and all the equivalences use symmetrical monoidal natural transformations. 
All horizontal lines are inclusion functors. Almost every category comes in two forms: all the possible morphisms and the subcategory of total morphisms. The diagram has a central line that consists of different types of functions that our models try to mimic. There are three spokes coming out of that line. These correspond to three types of models of computation: (i) the top spoke corresponds to models that manipulate strings; (ii) the lower right spoke corresponds to models that manipulate natural numbers; and (iii) the lower left spoke corresponds to models that manipulate bits. 

In all these categories, the composition corresponds to sequential processing (that is, performing one process after another). The monoidal structure corresponds to parallel processing. The symmetric monoidal structure corresponds to the fact that the order of parallel processing is easily interchanged. 

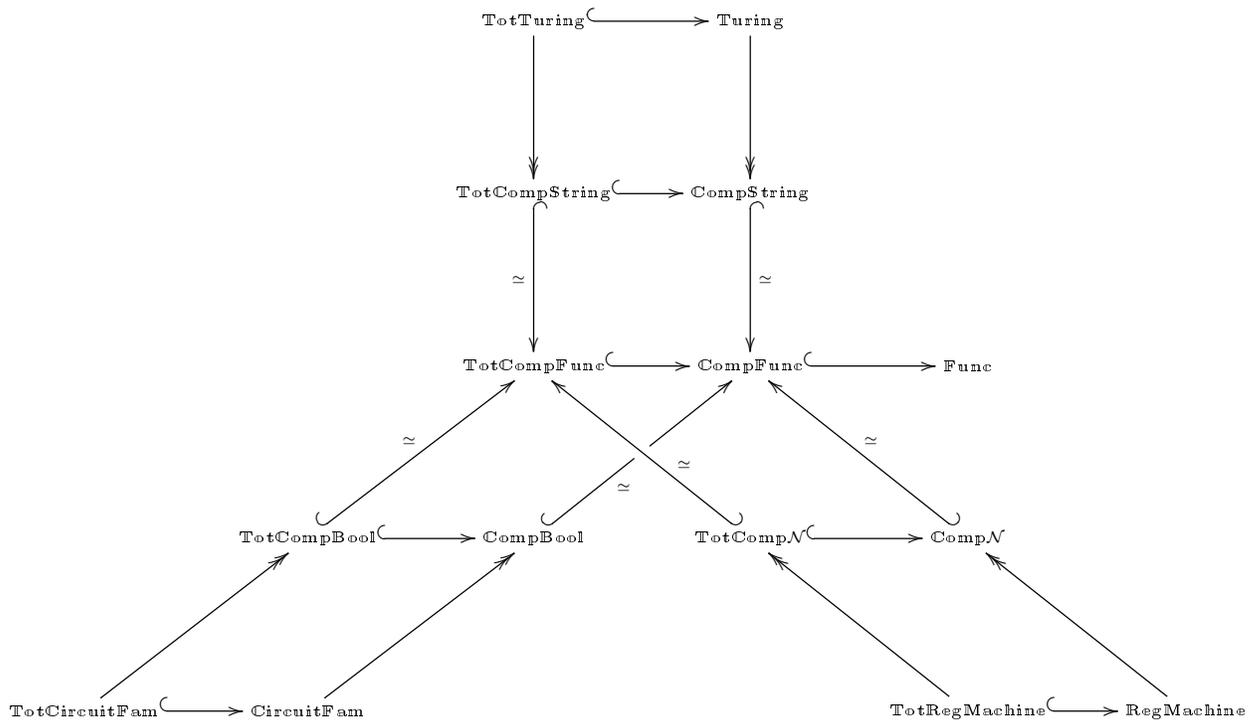
\begin{figure}

\begin{tiny}\xymatrix{
&&\TotTuring\ar@{->>}[dd] \ar@{^{(}->}[r]& \Turing\ar@{->>}[dd]\\
\\
&&\TotCompString\ar@{^{(}->}[dd]_\simeq  \ar@{^{(}->}[r]& \CompString\ar@{^{(}->}[dd]^\simeq\\
\\
&&\TotCompFunc  \ar@{^{(}->}[r]& \CompFunc \ar@{^{(}->}[r]&\Func\\
\\
&\TotBoolComp\ar@{^{(}->}[uur]^\simeq  \ar@{^{(}->}[r]&\BoolComp\ar@{^{(}->}[uur]|\hole_<(.3)\simeq &\TotCompN\ar@{_{(}->}[uul]_<(.3)\simeq  \ar@{^{(}->}[r]&\CompN\ar@{_{(}->}[uul]_\simeq \\
\\
\TotCircuitFam\ar@{->>}[ruu] \ar@{^{(}->}[r]&\CircuitFam\ar@{->>}[ruu]&&&\TotRegMachine\ar@{->>}[uul] \ar@{^{(}->}[r]&\RegMachine\ar@{->>}[uul]
}
\end{tiny}
\caption{``The Big Picture'' of models of computation}\label{fig:modelsofcomputation}
\end{figure}

The central focus of this mini-course  is the middle line of The Big Picture.
\begin{defi} The category $\Func$ consists of all functions from sequences of types to sequences of types. The objects are sequences of types such as $Nat \times String \times Nat \times  Bool$ or $Float \times Char \times Nat \times Integer$. We will denote a sequence of types as $Seq$. The morphisms in $\Func$ from $Seq$ to $Seq'$ are all functions that have inputs from type $Seq$ and outputs of type $Seq'.$ We permit all types of functions including partial functions and functions that computers cannot mimic. The identity functions are obvious. Composition in the category is simply function composition. The monoidal structure on objects is concatenation of sequences of types. Given $f \colon Seq_1 \lra Seq_2$ and $g\colon Seq_3 \lra Seq_4$, their tensor product is $(f \otimes g) \colon (Seq_1 \times Seq_3) \lra (Seq_2 \times Seq_4)$ which corresponds to performing both functions in parallel. The symmetric monoidal structure comes from the trivial function that swaps sequences of types, i.e., $tw \colon Seq \times Seq' \lra Seq' \times Seq$. We leave the details for the reader.  

The category $\CompFunc$  is a subcategory of $\Func$ which has the same objects. The morphisms of this subcategory are functions that a computer can mimic. Partial functions are permitted in this subcategory. A computer can mimic a partial function if for any input for which there is an output, the computer will give that output, and if there is no output, the computer will not output anything or go into an infinite loop. 

There is a further subcategory $\TotCompFunc$ which contains all the total computable functions. These are functions that for every input there is an output and a computer can mimic the function.  There are obvious inclusion functors 
\be \xymatrix{\TotCompFunc  \ar@{^{(}->}[r]& \CompFunc \ar@{^{(}->}[r]&\Func .\\}\label{diag:inclusionsoffunc}\ee
which are the identity on objects.
\end{defi}

In the definition we saw the phrase ``functions that a computer can mimic.'' The obvious question is what type of computer are we discussing? What computer process is legitimate? This part of the mini-course will give several answers to that question.

\subsection{Manipulating Strings: Top Spoke.}

Let us go up through the top spoke of The Big Picture.
\begin{defi}The category $\CompString$ is a subcategory of $\CompFunc$. The objects are sequences of only $String$ types. We do not permit any other types.  The objects are $String^0=\ast$ (which is the terminal type), $String^1=String$, $String^2=String \times String$, $String^3=String \times String \times String$, $\dots$.  The morphisms of this category 
are computable functions between sequences of $String$ types. There are partial functions in this category. The symmetric monoidal structure is similar to $\Func$.

The subcategory $\TotCompString$ has the same objects as $\CompString$ but with only total computable string functions. There is an obvious inclusion functor $\TotCompString \longhookrightarrow \CompString$. 
\end{defi}

There is an inclusion functor $\CompString \longhookrightarrow \CompFunc$  that takes $Strings^n$ in $\CompString$ to the same object in $\CompFunc$. This functor is more than an inclusion functor. 
\begin{theo} \label{theo:encodingtypes}
The inclusion functor $Inc\colon \CompString \longhookrightarrow \CompFunc$  is an equivalence of categories.  
\end{theo}
\begin{proof}
Every computable string function in $\CompString$ goes to a unique computable function in $\CompFunc$ so the inclusion functor is full and faithful. What remains to be shown is that the inclusion functor is essentially surjective. That means, given any sequence of types in $\CompFunc$, say $Seq$ there is some $n$ and a computable isomorphism $enc\colon Seq \lra String^n$ that encodes the data of type $Seq$ into data of type $String^n$. Every programmer knows how to encode one data type as another. This encoding is an isomorphism because we need to be able to uniquely decode the encoding. 

We are really describing a  functor $F\colon \CompFunc \lra \CompString$. The types in $\CompFunc$ are encoded as a sequence of strings and the morphisms are encoded as functions between sequences of strings. $F\circ Inc= Id_\CompString$ because products of strings are encoded as themselves and $Inc \circ F \cong Id_\CompFunc$. It is important to point out that there is nothing universal about any encoding. There might be many such encodings. However they are all isomorphic to each other.
\end{proof}

There is a similar equivalence $\TotCompString \longhookrightarrow \TotCompFunc$. 

\vspace{.5in}

Let us continue up the top spoke of The Big Picture. In the 1930's, Alan Turing wondered about the formal definition of a computation. He came up with a model we now call a {\bf Turing machine} which manipulate strings. Turing based his work on the analogy that mathematicians do computation. They manipulate the symbols of mathematics in different ways when they are in different states. For example, if a mathematician sees the statement $x\times (y + z)$ and is in the distributive state, she will then cross out that statement and write $(x \times y) + (x \times z)$. In the same way, a Turing machine has a finite set of states that describe what actions the machine should perform. Just as a mathematician writes his calculation on a piece of paper, so too, a Turing machine performs its calculations on paper.
Turing was inspired by ticker tape machines and typewriter ribbons to define his paper as a thin tape that can only have one character per space at a time. The machine will have several tapes that are used for input, a tape for working out calculations, and several tapes for output. For every tape, there will be an arm of the Turing machine that will be able to see what is on the tape, change one symbol for another, and move to the right or the left of that symbol. A typical rule of a Turing machine will say something like ``If the machine is in state $q_{32}$ and it sees symbol $x_1$ on the first tape and $x_2$ on the second tape, ... and the symbol $x_n$ in the $n$th tape, then change to state $q_{51}$, make the symbol in the first tape to a $y_1$ and the symbol in the second tape to $y_2$,... the symbol in the $n$th tape into a $y_n$, also move to the left in the first tape, the right in the second tape, ... , the right in the $n$th tape.'' 
In symbols we can write this as 
 \be \label{diag:TurDef} \delta(q_{32}, x_1,x_2, \ldots,x_n)  = (q_{51}, y_1,y_2, \ldots,y_n, L, R, R, L, \ldots, R).\ee      

There is an obvious question that we left out. How long should the paper be? Turing realized that if you limit the size of the paper, then you will only be able mimic certain less complicated functions. Since Turing was only interested in whether or not a function was computable, and not whether or not it it was computable with a certain sized paper, Turing insisted that the paper be infinitely long. There was no bound on how much calculation can be done. It was only thirty years later that theoretical computer scientists started being concerned with how much space resources are needed to compute certain functions. (We will see more of this in the complexity theory section of this mini-course.)

The type of Turing machines we will deal with will have several input tapes and several output tapes and another work tape which will contain all the calculations. We might envision the Turing machine as Figure \ref{pic:turingmachine}.

\begin{figure}[h!]
\centering
 \includegraphics[width=17cm, height=15cm]{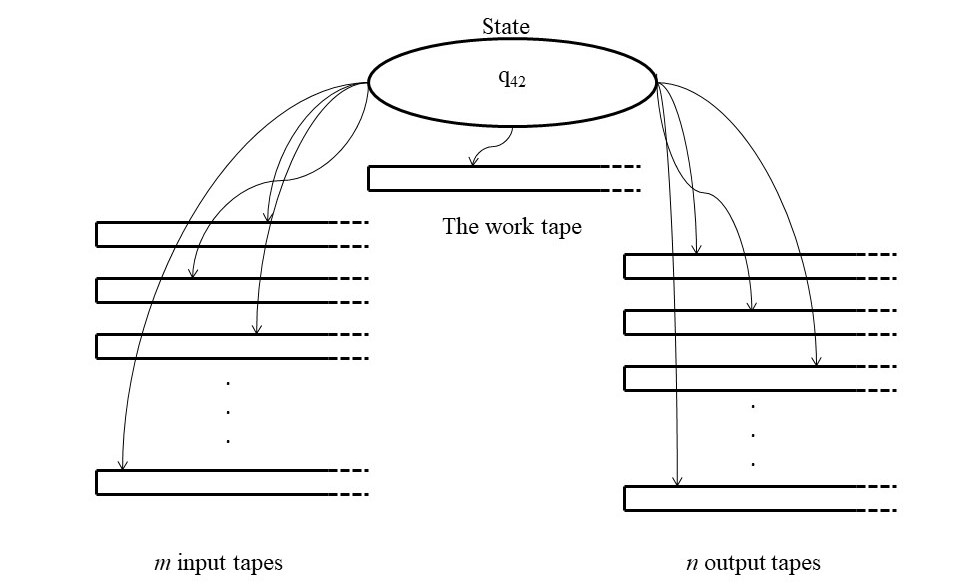}
  \caption{A Turing Machine} \label{pic:turingmachine}
\end{figure}

A computation occurs when data is put on the input tapes and the Turing machine is in a special starting state. The  machine then follows the rules on how to manipulate the input strings, compute on the work tape, and write on the output tape. It will go through many states manipulating the strings. There are two possible outcomes that can happen with this process: (i) the Turing machine can come to a certain state with symbols on its tapes for which there is no further rule. The machine then halts. Or (ii) the Turing machine continues forever in an infinite loop.  

Let us put all these machines in one category.
\begin{defi}
The category $\Turing$ consists of all the Turing machines. The objects are the natural numbers. The set $Hom_\Turing(m,n)$ consists of all Turing machines with $m$ input tapes and $n$ output tapes. The Turing machines compose in the obvious way. If $T\colon m \lra n$ and $T'\colon n \lra p$ are Turing machines, then $T'\circ T \colon m \lra n \lra p$ will be a Turing machine. The output tapes of $T$ become the input tapes of $T'$. After the $T$ machine halts, the machine will go to the start state of the $T'$ machine. If the $T$ machine does not halt, the $T'$ machine never even begins. The monoidal structure on the objects is addition of natural numbers. The monoidal structure for morphisms $T\colon m \lra n$ and $T'\colon m' \lra n'$ is $T \otimes T' \colon m+m' \lra n+n'$. The Turing machine $T \otimes T'$ has to be defined by doing both processes at one time. The set of states is the product of the two sets of states. At each point of time, this Turing machine does what both machines would do. 

There are many Turing machines that for certain input do not stop but go into an infinite loop. Others halt on all inputs. The subcategory $\TotTuring$ consists of the Turing machines that halt on every input. There is an obvious inclusion functor $\TotTuring \longhookrightarrow \Turing$. 
\end{defi}

There is one problem: I lied to you.
\begin{tech}\label{tech:TuringNotCat}
 $\Turing$ is not really a category. While it is a directed graph with a well defined associative composition, and there is a identity Turing machine $Id_n\colon n \lra n$ that takes all the data on the input tapes to the output tapes, there is a problem with composition with this identity. The composition of any Turing machine with such an identity Turing machine produces the correct function, but it is not the same Turing machine as the original Turing machine, i.e. $T\circ Id_n \not = T$.  
There are different ways of dealing with this problem: (i) We can be very careful in defining the composition of Turing machines. (ii) We can talk about equivalence classes of Turing machines and in this case $T\circ Id_n \sim T$. Or (iii) we can begin to define new structures called ``almost-categories'' with symmetric monoidal structures. 

In \cite{mydefalg} this problem is taken very seriously and method (ii) is used to deal with it. Different equivalence relations are discussed in \cite{mygalois}. However, for the purposes needed here, we will call it a category, but be aware of the problem. We will be careful with what we say.  
\end{tech}

Every Turing machine describes a function. By looking at all the input and its outputs, we are defining a function. This is actually a functor $\Turing \lra \CompString$. This functor  will take object $m$ to $Strings^m$ and a Turing machine with $m$ inputs and $n$ outputs will go to a function from $String^m$ to $String^n$.  

It is believed that one can go the other way. Given any computable function, we can find a Turing machine that computes it. This is the content of the  
{\bf Church-Turing thesis} which says that any computable function can be mimicked by a Turing machine. In our categorical language this means that the functor $\Turing \lra \CompString$ is full. This is called a ``thesis'' rather than a ``theorem'' because it has not been proven and probably cannot be proven. The reason for the hardship is that there is no perfect definition of what it means to be computable function. How can we prove that every computable function can be mimicked by a Turing machine when we cannot give an exact characterization of what we mean by computable function? In fact, some people define a computable function as a function which can be mimicked by a Turing machine. If we take that as a definition, then the Church-Turing thesis is true but has absolutely no content. Be that as it may, most people take the Church-Turing thesis to be true. Turing machines have been around since the 1930's and no one has found a computable function that a Turing machine cannot mimic. Another way to see this is to realize that if a Turing machine {\it cannot} mimic some function then the function is not computable and no computer can mimic it. In the next part of this mini-course we will describe functions that cannot be mimicked by a Turing machine and hence cannot be mimicked by any computer.  

There is an intimate relationship between computation and logic. We shall describe this relationship with a (symmetric monoidal) functor from the (symmetric monoidal) category of Turing machines to a (symmetric monoidal) category of logical formulas. That is, we will formulate a category $\Logic$ and describe a functor 
\be L \colon \Turing \lra \Logic.\ee
There are many ways of describing the collection of logical formulas. We will describe $\Logic$ so that it fits nicely with the category of Turing machines. 

Let us first see what we need from the functor $L$. Logical formulas that describe Turing machines will need three types of variables:
\begin{itemize}
\item There are variables to describe the contents of the tapes. A typical variable will be $C^z(t, i,j,k)$ where $z \in \{i, w, o \}$ corresponding to input tape, work tape and output tape. $C^z(t, i,j,k)$ is true iff at time $t$, on the $i$th $z$ tape, the $j$th position contains symbol $k$. $t$ can be any non-negative integer, $i\in\{1, 2, \ldots , m\}$ where $m$ is the number of tapes of type $z$ (Since there is only one work tape, if $z=w$ then $i=1$). $j$ is any positive integer. $k\in \{ 1, 2 ,\ldots , |\Sigma|\}$ where $\Sigma$ is the alphabet of the Turing machine.  
\item There are variables to describe the position of the pointers to the tapes. A typical variable will be $P^z(t, i,j)$ where $z \in \{i, w, o \}$ corresponding to input tape, work tape and output tape. $P^z(t, i,j)$ is true iff at time $t$, on the $i$th $z$ tape, the pointer is pointing to the $j$th position. The number of variables is similar to $C^z$.
\item There are variables to describe what state the Turing machine is in. $Q(t,q)$ is true iff at time $t$ the Turing machine is in state $q$. $q \in \{1,2,\ldots,  |Q| \}$ where $Q$ is the set of states of the Turing machine.  
\end{itemize}

Now that we have variables, let us deal with the formulas. We will not go into all the details of all the formulas, but we will give a sampling of the types of formulas. 
\begin{itemize}
\item For every tape, at every time click, at each position (up to a point), there exists something in contents of the tape:
\be C^z(t,i,j,1) \lor C^z(t,i,j,2) \lor  \cdots \lor C^z(t,i,j,|\Sigma |) \ee
for appropriate $z, t, i$ and $j$. There are similar formula for $P^z$ and $Q$.   
\item For every tape, at every time click, the Turing machine is not pointing to more than one position on the tape:
\be [j\neq j'] \lra [P^z(t, i,j) = \neg P^z(t, i,j')]\ee
For appropriate $z,t$ and $i$. There are similar formulas for $C^z$ and $Q$.
\item At time $t=0$ the Turing machine is in state 1 and all the pointers are pointing to position 1:
\be Q(0,1) \land P^z(0, i, 1) \ee
for the appropriate $z$ and $i$.
\end{itemize}

The most important logical formulas will correspond to the instructions --- or program ---  of the Turing machine. Let us start formalizing Equation \ref{diag:TurDef}. To make it easier, we will just assume $n$ tapes of type $z$.
\be [Q(t,32)\ee
\be \land P^z(t,1,j_1)\land P^z(2,2,j_2)\land\cdots \land P^z(t,n,j_n)\ee
\be \land C^z(t,1,j_1, x_1)\land C^z(t,2,j_2,x_2))\land \cdots \land C^z(t,n,j_n,x_n)]\ee
\be \qquad \lra \qquad \ee
\be [Q(t+1,51)\ee
\be \land C^z(t+1,1,j_1, y_1)\land C^zi(t+1,2,j_2,y_2)\land \cdots \land C^z(t+1,n,j_n,y_n)\ee
\be \land P^z(t+1,1,j_1-1)\land P^z(t+1,2,j_2+1)\land \cdots \land P^z(t+1,n,j_n+1)]\ee

The first line tells you the state. The second line tells the position of all the pointers. The third line tells the content. If the first three lines are true, then the implication on line four tells that the new state ( at time $t+1$ ) should be on line five. The new content is on line six. The new position is on line seven. We have to do this for each rule. 
We can go further with the details but we are going to follow the motto that ``just because something could be formalized does not mean it should be formalized!''

For any Turing machine $T$, we will have many different logical formulas described above. We will call the conjunction of all these logical formulas $L(T)$. This large logical formula describes the potential computation of $T$. There is, however something missing: input. Let $x$ stand for the input, then we can conjunction $L(T)$ with logical formulas that say that at time $t=0$ input $x$ will be found on the input tapes. We denote this larger logical formulas as $L(T)[x]$. Depending on the input, $L(T)[x]$ might be satisfiable and it might not be. This will correspond to whether or not the computation will halt or not.   

Now that we understand what we want from the functor $L \colon \Turing \lra \Logic$, let us describe the category $\Logic$.
\begin{defi} 
The objects of $\Logic$ will be --- just like $\Turing$ --- the natural numbers. The morphisms from $m$ to $n$ will be logical formulas that use $m$ variables of the form $C^i$ and $P^i$ and n variables of the form $C^o$ and $P^o$. The composition of the logical formula $f\colon m \lra n$ with logical formula $f'\colon n \lra p$ will be the logical formula 
\be f \land [C^o(1)=C'^i(1)]\land  [P^o(1)=P'^i(1)]\land [C^o(2)=C'^i(2)]\land  [P^o(2)=P'^i(2)]\land \ee 
\be \cdots \land [C^o(n)=C'^i(n)]\land  [P^o(n)=P'^i(n)]\land f' \ee
i.e., the conjunction of the formulas and the setting of output variables of formula to have the same values as the input variables of the second formula. The monoidal structure is the addition of the natural numbers and the appropriate conjunction of formulas for the morphismism. There is also a symetric monoidal structure to $\Logic$. The functor $L$ is then a symmetric monoidal functor. 

We freely admit that definition of $\Logic$ is {\it ad hoc}. There are many ways to classify logical formulas. We choose this one because it fits nicely with the symmetric monoidal structure of $\Turing$. 
\end{defi}

\vspace{.5in}

We will now discuss the bottom two spokes of The Big Picture. We note that the rest of this mini-course will be developed with the categories and functors from the top spoke alone. The categories in the bottom two spokes will barely be mentioned again. The ideas and theorems of theoretical computer science could be said using the language of any of the spokes. We choose Turing machines for historical reasons.  If you are not interested in other models of computation, you can skip them. We include them only because they are discussed in many textbooks on theoretical computer science.

\subsection{Manipulating Natural Numbers: Lower Right Spoke}

While Turing thought of a computation as manipulating strings, others such as Joachim Lambek, Marvin Minsky, John C. Shepherdson, and Hao Wang thought of a computation as something to do with manipulating natural numbers. They dealt with functions whose input and output were sequences of natural numbers.

This brings us to define the following categories.
\begin{defi}
The objects of the category $\CompN$ are types $Nat^0= \ast$, $Nat^1=Nat$, $Nat^2=Nat \times Nat$, $Nat^3=Nat \times Nat \times Nat$, $\ldots$. The morphisms are all computable functions from powers of natural numbers to powers of natural numbers (including partial functions). The symmetric monoidal category structure is the same as that of $\Func$. 

There is a subcategory $\TotCompN$ with the same objects but contains only total computable functions. There is an obvious inclusion $ \TotCompN \longhookrightarrow \CompN$.
\end{defi}

There is an inclusion $Inc \colon \CompN \longhookrightarrow \CompFunc$ that takes $Nat^m$ to $Nat^m$. Just like we can encode any sequence of types as strings, so too, we can encode any sequence of types as natural numbers. From this we get the analogy of 
Theorem \ref{theo:encodingtypes} that says that this inclusion is an equivalence of categories. Similarly, the inclusion 
$Inc \colon \TotCompN \longhookrightarrow \TotCompFunc$ is also an equivalence of categories. 

Just like a Turing machine is a method to manipulate strings, a {\bf register machine} is a method of manipulating natural numbers. 
What is a register machine? They are basically programs in a very simple programing language. These programs have three different types of variables. There are $X_1, X_2, X_3, \ldots$ which are called input variables,
$Y_1, Y_2, Y_3, \ldots$ which are called output variables and $W_1, W_2, W_3, \ldots$ which are called work variables. In a register machine one is permitted the following type of operations on any variable $Z$:
\begin{itemize}
\item $Z=0$
\item $Z=Z+1$
\item If $Z = 0$ goto $L$
\end{itemize}
where $L$ is a label for some line number. A program is a list of such statements for several variables. The register machine usually starts with the input variables initialized to the inputs. The machine then follows the program. Each variable is to be thought of as a computer register that holds a natural number. The values in the output variables at the end of an execution are the output of the function. There exists certain register machines for which some of the input causes the machine to go into an infinite loop and have no output values. Other register machines halt for any input.   

Let us put all the register machines into a category.
\begin{defi}
The objects of the category $\RegMachine$ are the natural numbers. The morphisms from $m$ to $n$ are all register machines with $m$ input variables and $n$ output variables. There are identity register machines that do nothing but take $m$ inputs and put them into $m$ outputs without changing the values. Composition is not hard to define. Basically one program is tagged onto the end of another program. Output variables of the first program must be set equal to input variables of the second program. Labels and variable names must be changed so that there is no overlap. All this can be formalized with a little thought. The symmetric monoidal structure is all very similar to the structure in $\Turing$.  

There is a subcategory $\TotRegMachine$ whose objects are also the natural numbers and whose morphisms are total register machines, i.e, they have values for all input. There is an obvious inclusion 
$ \TotRegMachine \longhookrightarrow\RegMachine$.

(These ``categories'' have the same problem as $\Turing$ and $\TotTuring$ that we discussed in Technical Point \ref{tech:TuringNotCat}.)
\end{defi}

There is a functor $\RegMachine \lra \CompN$ that takes a register machine to the function it describes. The belief that every computable function on natural numbers can be mimicked by a register machine means that this functor is full. This is simply another statement of the Church Turing thesis that we saw earlier. There is a similar full functor $\TotRegMachine \lra \TotCompN$.

\vspace{.5in}
Besides for register machines there is another way to describe the category $\CompN$. The morphisms can be generated from special morphisms using particular types of generating operations. The special morphisms in the category $\CompN$ are called {\bf basic functions}:
\begin{itemize}
\item The {\bf zero function} $z\colon Nat \lra Nat$ which is defined for all $n$ as $z(n)=0$.
\item The {\bf successor function} $s\colon Nat \lra Nat$ which is defined for all $n$ as $s(n)=n+1$. 
\item The {\bf projections functions} for each $n$ and for each $i\leq n$, $\pi^n_i\colon Nat^n \lra Nat$ which is defined as $\pi^n_i(x_1, x_2, x_3 , \ldots, x_n)=x_i$. 
\end{itemize}
These morphisms are clearly computable and hence in $\CompN$. 

There are three operations on morphisms in $\CompN$:
\begin{itemize}
\item The {\bf composition operation}: given $f_1\colon Nat^m \lra Nat$, $f_2\colon Nat^m \lra Nat$, $\ldots$, $f_n\colon Nat^m \lra Nat$ and $g\colon Nat^n \lra Nat$, there is a function $h\colon Nat^m \lra Nat$ defined as 
\be h(x_1, x_2,  \ldots, x_m)=g(f_1(x_1, x_2, \ldots, x_m), f_2(x_1, x_2, \ldots, x_m), \ldots, f_n(x_1, x_2, \ldots, x_m))\ee
\item The {\bf recursion operation}: given $f\colon Nat^m \lra Nat$ and $g\colon Nat^{m+2} \lra Nat$ there is a function $h\colon Nat^{m+1} \lra Nat$ defined as  
\begin{align}
h(x_1, x_2, x_3 , \ldots, x_m, 0 )&=f(x_1, x_2, x_3 , \ldots, x_m)\\
h(x_1, x_2, x_3 , \ldots, x_m, n+1)&= g(x_1, x_2, x_3 , \ldots, x_m, n, h(x_1, x_2, x_3 , \ldots, x_m, n))  
\end{align}
\item The {\bf $\mu$-minimization operation}: given $f\colon Nat^{m+1} \lra Nat$ there is a function $h\colon Nat^m \lra Nat$ that is defined as follows
\begin{align} 
h(x_1, x_2, x_3 , \ldots, x_m) &= \mbox{ the smallest number $y$ such that }  f(x_1, x_2, x_3 , \ldots, x_m,y)=0\\
							&=\mu_y [f(x_1, x_2, x_3 , \ldots, x_m,y)=0]
\end{align}
If no such $y$ exists, no value is returned for $h$ with those inputs. 
\end{itemize}

One can generate morphisms in $\CompN$ in the following manner. Start with the basic functions and then perform these three operations on them. Add the resulting morphisms of these operations to the set of morphisms that you perform the operations. Continue generating morphisms in this manner.  The conclusion is stated in the following theorem.
\begin{theo}
All the morphisms in $\CompN$ are generated by the operations of composition, recursion and minimization starting from the basic functions. \end{theo}
This result is proven in Chapter 3 of \cite{davis}, Chapter 2 of \cite{cutland}, and Chapter 6 of \cite{boolos}. 

It is interesting to examine which of these morphisms are in $\TotCompN$. All the basic functions are in $\TotCompN$. Notice that if the $f_i$s and $g$ of the composition definition are in $\TotCompN$ then so is $h$, i.e., $\TotCompN$ is closed under the composition operation. 
$\TotCompN$ is also closed under the recursion operation. In contrast, $\TotCompN$ is not closed under the $\mu$-minimization operation. That is, there could be an $f$ and a $x_1, x_2, x_3 , \ldots, x_m$ such that there does not exist a $y$ with $f(x_1, x_2, x_3 , \ldots, x_m,y)=0$. In that case $h(x_1, x_2, x_3 , \ldots, x_m)$ is not defined. $h$ is then a partial function and hence a morphism in $\CompN$ but not in $\TotCompN$. 

When the $\mu$-minimization operation is omitted we have an interesting class of total computable functions. 
\begin{defi} The set of morphisms  generated by the operations of composition and recursion from the basic functions are called {\bf primitive recursive functions}. There is a subcategory $\PRCompN$ of $\TotCompN$ which has the same objects (products of natural numbers types) and its morphisms are the primitive recursive function. There are obvious inclusions \be\PRCompN \longhookrightarrow \TotCompN \longhookrightarrow \CompN.\ee 
\end{defi}

We close our discussion of primitive recursive function with an interesting historical vignette. Primitive recursive functions were defined by David Hilbert. He believed that this category of functions was what was meant by a (total) computable function. Hilbert had a student named Wilhelm Ackermann who showed that the class of all primitive recursive functions does not contain all total computable functions. That is, there is a morphism in $\TotCompN$ called {the \bf Ackermann function} $A\colon Nat \times Nat \lra Nat$ that is computable but is not primitive recursive. $A$ is defined as follows:   
\be
  A (m, n) = \left\{
      \begin{array}{ll}
        n + 1          &:  \mathrm{if}\ m = 0 \\
        A (m - 1, 1)   &:  \mathrm{if }\ m > 0 \mathrm{\ and\ } n = 0 \\
        A (m - 1, A (m, n - 1))   &: \mathrm{if}\ m > 0 \mathrm{\  and\  } n > 0 
      \end{array}
              \right.
\ee

The fact that the Ackermann function is not  primitive recursive can be seen in Section 4.9 of \cite{davis}. (There is a lot of fun in programming the Ackermann function and determining its values. Try to get your computer to find the value of $f(4,4)$.)

\subsection{Manipulating Bits: Lower-Left Spoke.}

While one can think of a computation as manipulating strings or numbers, the most obvious way to think of a computation is as a process that  manipulates bits. After all, every modern computer is implemented by manipulating bits. 

We need a type which we did not need before and we were not explicit about it. For every type $T$, there is a type $T^*$ which is finite strings of type $T$. In particular, the type $Bool^*$ is the type of strings of Boolean type, that is,  strings of $0$'s and $1$'s. 

\begin{defi}
The category $\BoolComp$ has powers of $Bool^*$ type as objects. A typical objects is $(Bool^*)^n$. The morphisms in this category are computable functions whose input and output are powers of strings of Boolean types. These functions might be partial functions.

There is a a subcategory $\TotBoolComp$ that has the same objects but whose morphisms are total computable functions. 
There is an obvious inclusion $\TotBoolComp\longhookrightarrow\BoolComp$
\end{defi}

There is an inclusion function $Inc \colon \BoolComp \longhookrightarrow \CompFunc$ such that $Inc((Bool^*)^n)=(Bool^*)^n$ that is full and faithful. However since any sequence of types can be encoded as Boolean variables, we have (similar to Theorem \ref{theo:encodingtypes}) that this inclusion function is an equivalence. Similarly, the functor $Inc \colon \TotBoolComp \longhookrightarrow \TotCompFunc$ is an equivalence.  

What type of physical devices mimic Boolean functions? Boolean circuits. In order for our circuits to be as powerful as the other models of computation, our circuits will need families of inputs and outputs. Let us put all such circuits together in one category called $\CircuitFam$. 
\begin{defi}
The objects of the category $\Circuit$ are the finite sequences of natural numbers, e.g., $5,7,12,23,0,13$. We will denote a typical object as $x_1,x_2, \ldots x_m$. The set of morphisms from $i_1, i_2, \ldots i_m$ to $o_1, o_2, \ldots o_n$ is the set of logical circuits (built from ANDs, ORs, NOTs, NANDs, NOR, etc.) with $m$ families of inputs and $n$ families of outputs. The $t$th family of inputs will have $i_t$ wires. The $t$th family of outputs will have $o_t$ wires. Such a circuit will be denoted as \be C^{i_1, i_2, \ldots i_m}_{o_1, o_2, \ldots o_n}\ee and will be drawn as follows
\be
\Qcircuit @C=2em @R=2em {
&\ustick{i_1}&{/} \qw & \multigate{3}{\qquad C\qquad } & \qw {/} &\ustick{o_1}\qw&\\
&\ustick{i_2}&{/} \qw & \ghost{\qquad C\qquad} & \qw {/} &\ustick{o_2}\qw   \\
&\ustick{\vdots}&{/} \qw & \ghost{\qquad C\qquad} & \qw {/} &\ustick{\vdots}\qw \\
&\ustick{i_m}&{/} \qw & \ghost{\qquad C\qquad} & \qw {/} &\ustick{o_n}\qw\\
}
\ee
The point is that with this formalism we can discuss circuits with $m$ Boolean strings of any length as inputs and have $n$ Boolean strings of any length as output. We need one more requirement: the circuits must be able to be described a computer (this is to make sure that we are not talking about {\it all} functions.)  $\CircuitFam$ almost forms a category. The identity circuit is simply the correct number of plain wires without any gates.  Similar to the problem mentioned in Technical Point \ref{tech:TuringNotCat}, $\CircuitFam$ is not really a category. Composition of circuits is given as follows: if there are circuits 
\be C^{i_1, i_2, \ldots i_m}_{o_1, o_2, \ldots o_n} \mbox{              and            }  C'^{o_1, o_2, \ldots o_n}_{p_1, p_2, \ldots p_k}, \ee
then they can be combined by attaching the output of the first with the input of the second to form
\be C''^{i_1, i_2, \ldots i_m}_{p_1, p_2, \ldots p_k}. \ee We can draw this attached circuits as follows:
\be
\Qcircuit @C=2em @R=2em {
&\ustick{i_1}&{/} \qw & \multigate{3}{\qquad C\qquad } & \qw {/} &\ustick{o_1}\qw&{/} \qw & \multigate{3}{\qquad C'\qquad } & \qw {/} &\ustick{p_1}\qw&\\
&\ustick{i_2}&{/} \qw & \ghost{\qquad C\qquad} & \qw {/} &\ustick{o_2}\qw & {/} \qw &\ghost{\qquad C'\qquad} & \qw {/} &\ustick{p_2}\qw  \\
&\ustick{\vdots}&{/} \qw & \ghost{\qquad C\qquad} & \qw {/} &\ustick{\vdots}\qw &{/} \qw & \ghost{\qquad C'\qquad} & \qw {/} &\ustick{\vdots}\qw\\
&\ustick{i_m}&{/} \qw & \ghost{\qquad C\qquad} & \qw {/} &\ustick{o_n}\qw&{/} \qw & \ghost{\qquad C'\qquad} & \qw {/} &\ustick{p_k}\qw\\
}
\ee
Composition with the identity circuits does not change anything. Associativity of composition is straightforward.

There is a monoidal structure on $\CircuitFam$. The monoidal structure on the objects is simply composition of sequences of natural numbers. That is, 
\be i_1, i_2, \ldots i_m \otimes i'_1, i'_2, \ldots i'_{m'} = i_1, i_2, \ldots i_m,  i'_1, i'_2, \ldots i'_{m'}. \ee
The monoidal structure on circuits is given by placing the circuits in parallel. In pictures:
\be
\Qcircuit @C=2em @R=2em {
&\ustick{i_1}&{/} \qw & \multigate{3}{\qquad C\qquad } & \qw {/} &\ustick{o_1}\qw&\\
&\ustick{i_2}&{/} \qw & \ghost{\qquad C\qquad} & \qw {/} &\ustick{o_2}\qw   \\
&\ustick{\vdots}&{/} \qw & \ghost{\qquad C\qquad} & \qw {/} &\ustick{\vdots}\qw \\
&\ustick{i_m}&{/} \qw & \ghost{\qquad C\qquad} & \qw {/} &\ustick{o_n}\qw\\
&\ustick{j_1}&{/} \qw & \multigate{3}{\qquad C'\qquad } & \qw {/} &\ustick{k_1}\qw&\\
&\ustick{j_2}&{/} \qw & \ghost{\qquad C'\qquad} & \qw {/} &\ustick{k_2}\qw   \\
&\ustick{\vdots}&{/} \qw & \ghost{\qquad C'\qquad} & \qw {/} &\ustick{\vdots}\qw \\
&\ustick{j_{m'}}&{/} \qw & \ghost{\qquad C'\qquad} & \qw {/} &\ustick{k_{n'}}\qw\\
}
\ee

\noindent Formally, the monoidal structure on morphisms is given as follows:
\be C^{i_1, i_2, \ldots i_m}_{o_1, o_2, \ldots o_n} \otimes  C'^{j_1, j_2, \ldots j_{m'}}_{k_1, k_2, \ldots, k_{n'}} =  C''^{i_1, i_2, \ldots i_{m}, j_1, j_2, \ldots, j_{m'}}_{o_1, o_2, \ldots o_{n}, k_1, k_2, \ldots, k_{n'}}\ee
where $C''$ is just the circuit $C$ next to the circuit $C'$. The symmetric monoidal structure comes from twisting the wires across each other as follows:

\be \xymatrix@C=2em @R=0em{ 
\bullet \ar[rrrrddddddd]&&&&\bullet\\
\bullet \ar[rrrrddddddd]&&&& \bullet\\
\bullet \ar[rrrrddddddd]&&&& \bullet\\
\bullet \ar[rrrrddddddd]&&&& \bullet\\
           &&&& \bullet\\
\\
\\
\bullet \ar[rrrruuuuuuu]&&&&\bullet\\
\bullet \ar[rrrruuuuuuu]&&&& \bullet\\
\bullet \ar[rrrruuuuuuu]&&&& \bullet\\
\bullet \ar[rrrruuuuuuu]&&&& \bullet\\
\bullet \ar[rrrruuuuuuu]&&&& \\
}
\ee

There is a subcategory with the same objects and with only circuits that do not go into infinite loops called $\TotCircuitFam$. There is an obvious inclusion $\TotCircuitFam \longhookrightarrow\CircuitFam$
\end{defi}

 Notice that if a circuit does not have any feedback it is a total function. In contrast, if a function does have feedback, there might be some inputs that force the circuit into an infinite loop. (All computers have such circuits. Every time your computer is in an infinite loop, it is because there is feedback in the circuits.)

There is a symmetric monoidal functor $P\colon \CircuitFam \lra \BoolComp$ that takes the object $i_1, i_2, \ldots i_m$ to $(Bool^*)^m$ and takes a logical circuit with $m$ families of input wires and $n$ families of output wires to the Boolean computable function from $(Bool^*)^m$ to $(Bool^*)^n$ that the circuit describes. This functor really describes a congruence on the category $\Circuit$. The circuit $C\colon (i_1, i_2, \ldots i_m) \lra (o_1, o_2, \ldots o_n)$ is equivalent to $C'\colon (i_1, i_2, \ldots i_m) \lra (o_1, o_2, \ldots o_n)$ if $C$ and $C'$ describe the same Boolean function, i.e., $P(C)=P(C')$. It is well known that every Boolean function $h\colon (Bool^*)^m \lra (Bool^*)^n$ can be mimicked by a circuit with only NAND gates and the fanout operation. Another way to say this is that $P^{-1}(h)$ contains a circuit with only NAND gates and the fanout operation. Yet another way of saying this is that every circuit in $\Circuit$ is equivalent to a circuit with only NAND gates and fanout operations.

\vspace{.5in}
While all these different ways of characterizing computable functions are important, it would be somewhat redundant to prove any theorem in more than one way. They are all interchangeable. We choose to discuss the major ideas in theoretical computer science using one of the first models of computation, Turing machines. 

%
% Computability theory
%
\section{Computability Theory}
This part of the mini-course deals with the question of which functions are computable, and more interestingly, which functions are not computable. We will be stating our theorems in terms of the top spoke of The Big Picture. The categories $\TotCompString$ and $\CompString$ will not play a role here and so we will concentration on the following section of The Big Picture: 

\be \xymatrix{
\TotTuring\ar@{->>}[dd]_D \ar@{^{(}->}[r]& \Turing\ar@{->>}[dd]^Q\\
\\
\TotCompFunc  \ar@{^{(}->}[r]& \CompFunc \ar@{^{(}->}[r]&\Func.
}\ee

A large part of our discussion of computability theory is determining if a given morphism in $\Func$ is in $\CompFunc$ or in $\TotCompFunc$. Another way of looking at this is to consider the following functors 

\be \xymatrix{
\TotTuring\ar@{->}[ddr]_D \ar@{^{(}->}[rr]&& \Turing\ar@{->}[ddl]^Q\\
\\
&\Func\\
}\ee
and asking if a particular morphism in $\Func$ is in the image of $Q$ or in the image $D$ or neither.

In the literature there various names for morphisms in $\Func$.   
\begin{defi}
A function in $\CompFunc$ is obviously called {\bf computable} but it is also called {\bf Turing computable}. Sometimes a function that a Turing machine can mimic is also called {\bf solvable}.  There is special nomenclature for morphisms whose codomain is $Bool$. $f\colon Strings^n \lra Bool$ is called a {\bf decision problem}. An instance of the problem is put in the input of the function and the output is either true or false. If the decision problem is in $\TotCompFunc$, it is called {\bf recursive} or {\bf Turing-decidable} or simply {\bf decidable}. This means that there is a Turing machine that can give a ``true-false'' answer to this decision problem. If the decision problem is in $\CompFunc$, that is, in the image of $Q$, then we call that function {\bf recursively enumerable} or simply {\bf r.e.}, or  {\bf Turing-recognizable}, or {\bf semi-decidable}. That means that, for a given input to the function, a Turing machine can recognize when the answer is true. Decision problems will be a major focus in the coming pages.    
\end{defi}

\subsection{Turing's Halting Problem.}

While we are all familiar with many functions that computers can mimic, it is interesting to see functions that cannot be mimicked by any computer. First some preliminaries. A Turing machine is basically a set of rules about how to manipulate symbols within strings. The set of states, alphabet, and rules can all be encoded into a finite natural number. It is easy to see this because all the information can be encoded as a string and then we can encode the string into a binary number (using ASCII for example). That binary number is the number of the Turing machine. Another way to see this to is realize that every computer program can be stored as a sequence of zeros and ones which can be thought of as a binary number. The numbers for most Turing machines will be astronomically large, but that does not concern us. The point we are making is that every Turing machine can be encoded as a unique natural number. We do not make the requirement that the encoding should respect the composition of Turing machines or that the Turing machines with $m$ input tapes should be encoded as numbers less then Turing machines with $m+1$ input tapes. In other words we do not insist that the encoding respects the categorical structure of $\Turing$. An encoding can be thought of as an injective functor    
\be Enc\colon \coprod_m \coprod_n Hom_\Turing(m,n)  \lra d(\n)\ee
 from the discrete set of all Turing machines to the discrete set of natural numbers. Throughout our discussion, we are going to work with one encoding function $Enc$. If there is a Turing machine $T$ with $Enc(T)=y$ then we call $T$ ``Turing machine $y$''.  

We are interested when a Turing machine halts and when a Turing machine goes into an infinite loop. For many Turing machines, whether or not the Turing machine halts depends on its input. Let us simplify the problem by considering Turing machines that only accepts a single natural number as input.  There is a morphism in $\Func$ called $Halt\colon Nat \times Nat \lra Bool$ which is defined as follows
\be
  Halt(x, y) = \left\{
      \begin{array}{ll}
        1 &:  \mbox{if Turing machine $y$ on input $x$ halts.} \\
        0 &:  \mbox{if Turing machine $y$ on input $x$ does not halt. } \\
      \end{array}
              \right.
\ee
The {\bf Halting decision problem} asks if one is able to write a Turing machine to mimic the $Halt$ function. $Halt$ is a total function, but is it  in $\TotCompFunc$, $\CompFunc$ or only in $\Func$?  

\begin{theo}{\bf (Turing's undecidability of the halting problem.)}
The $Halt$ function is not in $\TotCompFunc$. That is, $Halt$ is not recursive (or Turing-decidable). \end{theo}
\begin{proof}
First some intuition. The proof is an example of a self-referential paradox. A famous self-referential paradox is the ``Liar paradox'' which says that the sentence ``This sentence is false'' is true if and only if it is false, i.e., it is a contradiction. Another example is G\"odel's incompleteness theorem which says that the mathematical statement ``This statement is unprovable'' is true but unprovable. Here, with the halting problem, the same type of self-referential statement is made as follows:
\begin{verse}
If there was a way to solve the halting problem, \\ then one can construct a program that performs the following task:\\ ``When you ask me if I will stop or go into an infinite loop,\\ then I will give the wrong answer.'' 
\end{verse}
Since computers do not give wrong answers, this program does not exist and hence the halting problem is unsolvable. 

In detail, the proof is a proof by contradiction. Assume (wrongly) that $Halt$ is in $\TotCompFunc$. Compose $Halt$ with the diagonal morphism $\Delta\colon Nat \lra Nat \times Nat$ and the ``partial Not'' morphism $ParNOT\colon Bool \lra Bool$. $\Delta$ is defined as $\Delta(n)=(n,n)$ and is in $\TotCompFunc$. $ParNOT$ is defined as follows

\be
  ParNOT(x) = \left\{
      \begin{array}{ll}
        1 &:  \mbox{if }x=0 \\
        \uparrow &:  \mbox{if }x=1. 
      \end{array}
              \right.
\ee  
where $\uparrow$ means go into an infinite loop. $ParNOT$ is not in $\TotCompFunc$ but is in $\CompFunc$. 
After composing as follows \be \xymatrix{ Nat\ar[rr]^\Delta\ar@/^.3in/[rrrrrr]^{Halt'}&&Nat \times Nat \ar[rr]^{Halt}&& Bool\ar[rr]^{ParNOT}&&Bool }\ee we obtain $Halt'$ which is in $\CompFunc$ since all of the morphisms it is composed of are in $\CompFunc$ (by assumption).

$Halt'$ is defined as 
\be
  Halt'(x) = \left\{
      \begin{array}{ll}
        1 &:  \mbox{if Turing machine $x$ on input $x$ does not halt.} \\
        \uparrow &:  \mbox{if Turing machine $x$ on input $x$ does halt.  } \\
      \end{array}
              \right.
\ee
Since $Halt'$ is in $\CompFunc$ there is some Turing machine, say $T_0$ that mimics $Halt'$. Suppose $Enc(T_0)=y_0.$ Let us ask $Halt'$ about itself by plugging in $y_0$ into $Halt'$. $Halt'(y_0)$ halts and outputs a $1$ if and only if Turing machine $y_0$ on input $y_0$ does not halt but goes into an infinite loop, i.e., $Halt'(y_0)=\uparrow$. This is a contradiction. The only thing we assumed was that $Halt$ was in $\TotCompFunc$. We conclude that $Halt$ is not in the subcategory $\TotCompFunc$ of $\Func$.    
\end{proof}

In contrast to the total function $Halt$, there is a partial halting function $ParHalt \colon Nat \times Nat \lra Bool$ defined as 
\be
  ParHalt(x, y) = \left\{
      \begin{array}{ll}
        1 &:  \mbox{if Turing machine $y$ on input $x$ halts.} \\
        0 \mbox{ or }\uparrow&:  \mbox{if Turing machine $y$ on input $x$ does not halt. } \\
      \end{array}
              \right.
\ee

\begin{theo} $ParHalt$ is in $\CompFunc$. 
\end{theo}
\begin{proof} We shall describe a Turing machine that can simulate Turing machine $y$ on input $x$. For a given $x$ and $y$, a Turing machine can look at the rules of Turing machine $y$ and simulate it on input $x$.  If Turing machine $y$ halts on input $x$ then output a 1. As long Turing machine $y$ on input $x$ does not halt, the simulation will go on. 
\end{proof}

We will prove that $ParHalt$ is not in $\TotCompFunc$ and find a morphism that is not even in $\CompFunc$. First a definition and theorem.  
\begin{defi}
The function $NOT\colon Bool \lra Bool$ is defined as $NOT(0)=1$ and $NOT(1)=0$ is obviously in $\TotCompFunc$.
\end{defi}
\begin{theo}\label{theo:fandfcomp}
Let $f \colon Seq \lra Bool$ be in $\CompFunc$ and let \be f^c=NOT \circ f\colon Seq \lra Bool \lra Bool.\ee Then $f$ is in $\TotCompFunc$ if and only if $f$ and $f^c$ are in $\CompFunc$.
\end{theo}
\begin{proof}
If $f$ is in $\TotCompFunc$, then $f$ is definitely in $\CompFunc$. Since $NOT\colon Bool \lra Bool$ is in $\TotCompFunc$  then $f^c=NOT \circ f$ is in $\TotCompFunc$ and hence in $\CompFunc$. 

The other direction of the proof is a little more complicated. One can gain intuition by looking at the three parts of Figure \ref{pic:deciderandrec}. In (i) we see a function that gives a true-false answer. One imagines the input entering on the left and either true or false is marked on the right. In (ii) we have a recognizer. The input enters on the left and true is answered on the right or there is no answer. Part (iii) of the Figure shows how one can build a decider by using two recognizers. The input is entered on the left and it goes into two recognizers. Both recognizers are executed in parallel. Since one of them is true, one of them will answer true.
 
\begin{figure}[h!]
\centering
 \includegraphics[width=10cm, height=9cm]{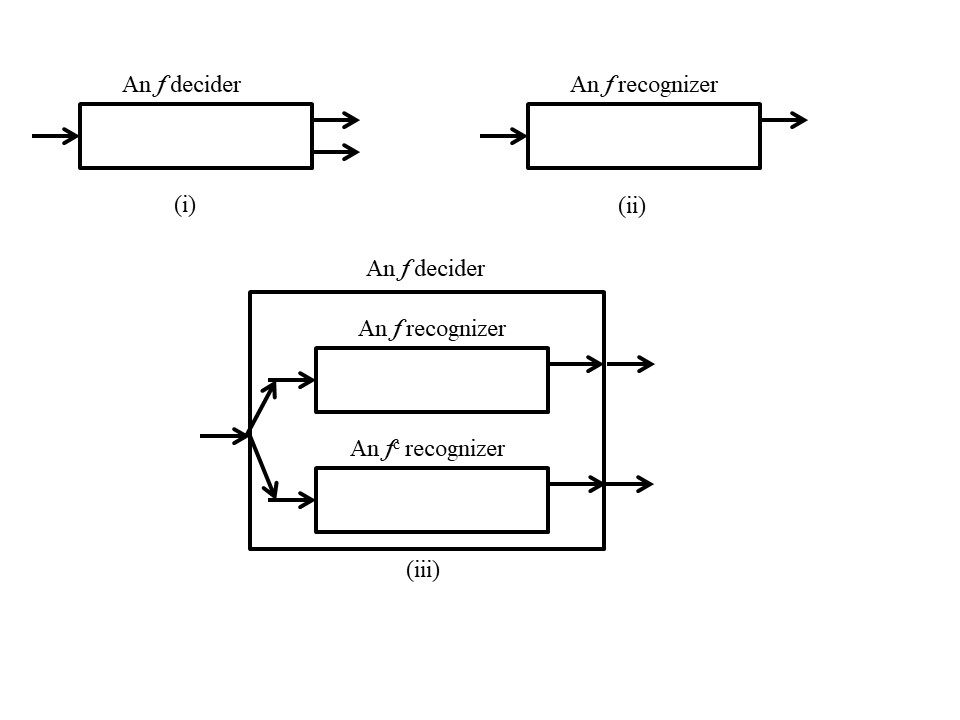}
  \caption{(i)a decider, (ii)a recognizer, and (iii)a decider built out of two recognizers} \label{pic:deciderandrec}
\end{figure}

In detail, assume that $f$ and $f^c$ are in $\CompFunc$.  The function that will be used to parallel processes those two functions at one time will be the function \be Parallel\colon Bool \times Bool \lra Bool \ee defined as 
\be
  Parallel(x, y) = \left\{
      \begin{array}{ll}
        1 &:  \mbox{if }x=1 \\
        0 &:  \mbox{if }y=1. 
      \end{array}
              \right.
\ee

The composition of the morphisms in $\CompFunc$ is given as follows. 
\be \xymatrix{Seq\ar[r]^{\Delta}&Seq\times Seq\ar[r]^{f\times f}\ar@/^.3in/[rr]^{f \times f^c}&Bool \times Bool\ar[r]^{id \times NOT}& Bool \times Bool\ar[rr]^{Parallel}&&Bool}\ee
This morphism is a total morphism and hence is in $\TotCompFunc$. 
\end{proof}

From Theorem \ref{theo:fandfcomp} we see that the partial halting function \be ParHalt^c=NOT \circ ParHalt \colon Nat \times Nat \lra Bool \lra Bool\ee  
\be
  ParHalt^c(x, y) = \left\{
      \begin{array}{ll}
        1 &: \mbox{if Turing machine $y$ on input $x$ does not halt.} \\
        0\mbox{ or } \uparrow&:  \mbox{if Turing machine $y$ on input $x$ halts. } \\
      \end{array}
              \right.
\ee
is not even in $\CompFunc$. It is a partial function that is in $\Func$ but no Turing machine can mimic it. 
This follows our intuition: while we can give a positive response when a Turing machine halts, how can we ever give a positive response that a Turing machine {\it will never} halt?

\subsection{Other unsolvable problems.}
While the halting problem is undecidable, it is just the beginning of the story. There are many other decision problems that are as hard or harder than the halting problem. They too are undecidable. But first we need a way of comparing decision problems.
\begin{defi}
Let $f\colon Seq \lra Bool$ and $g\colon Seq' \lra Bool$ be two functions in $\Func$. We say that $f$ is {\bf reducible} or {\bf reduces} to $g$ if there exists a $h\colon Seq \lra Seq'$ in $\TotCompFunc$ such that 
\be \xymatrix{Seq\ar[rr]^h \ar[ddr]_f&& Seq' \ar[ddl]^g\\ \\& Bool} \label{diag:reduction}\ee
commutes. We write this as $f \leq g$. If $f \leq g$ and $g \leq f$ then we write $f \equiv g$ and say they are both part of the same {\bf computability class}. 
\end{defi}
The way to think about such a reduction is that $h$ changes an $f$ input into a $g$ input. Letting $x$ be the input to $f$, the commuting triangle requirement means 
\be f(x) \mbox{ is true if and only if } g(h(x)) \mbox{ is true.}\ee
Notice that if there is a way to solve $g$ then there is definitely a way to solve $f$: simply use $h$ to change the input of $f$ into an input of $g$ and then solve it. The contrapositive of this statement is also important. If there is no way to solve $f$ then there is no way to solve $g$. Another way to say this is that $g$ is as hard or harder than $f$. 

A categorical way to examine reduciblity is to consider the following two functors
\be \xymatrix{ \TotCompFunc  \ar@{^{(}->}[rr]^{Inc}&& \Func && \one\ar[ll]_{Const_{Bool}} }\ee
where the functor on the left is the inclusion  and the functor on the right picks out the type $Bool$. We then take the comma category $(Inc, Const_{Bool})$. The objects of this category are morphisms in $\Func$ from some sequence of types to $Bool$. The morphisms are total computable functions that make Diagram \ref{diag:reduction} commute.

Let us use this notion of reducibility to prove that some morphisms are like $Halt$ and are not in $\CompFunc$. 
\begin{examp}
The {\bf nonempty program problem} asks if a given (number of a) Turing machine will have a nonempty domain. That is, will the given Turing machine accept any of its inputs. There is a morphism in $\Func$ called $Nonempty\colon Nat \lra Bool$ which is defined as follows
\be
  Nonempty(y) = \left\{
      \begin{array}{ll}
        1 &:  \mbox{if Turing machine $y$ has a nonempty domain} \\
        0 &:  \mbox{if Turing machine $y$ has empty domain. } \\
      \end{array}
              \right.
\ee

We show that the halting problem reduces to the nonempty program problem as in 
\be \xymatrix{Nat\times Nat\ar[rr]^{h} \ar[ddr]_{Halt}&& Nat \ar[ddl]^{Nonempty}\\ \\& Bool.}\ee
The total computable function $h$ is defined as follows for Turing machine $y$ and for input $x$. $h(x,y)=y'$ where $y'$ is the number of the Turing machine that performs the following task:

Turing machine $y'$: on input $w$
\begin{enumerate}
\item If $w\neq x$ reject. Stop.
\item If $w = x$ execute Turing machine $y$ on input $x$. If Turing machine $y$ accepts $x$, accept. Stop.
\end{enumerate}

Notice that Turing machine $y'$ depends on $x$ and $y$. Also notice that function $h$ is easily seen to be totally computable. That means that a program can easily compute $y'$ if it is given $x$ and $y$. Now consider Turing machine $y'$. It has at most one number in its domain. Only number $x$ can possibly be in its domain. Furthermore, $Nonempty(y')=1$ iff $Nonempty(h(x,y))=1$ iff the domain of Turing machine $y'$ is not empty iff $x$ is in the domain of Turing machine $y'$ iff  Turing machine $y$ accepts $x$ iff $Halt(x,y)=1$. But we already know that it is impossible to solve the halting problem. So it must be impossible to solve the nonempty problem.  
\end{examp}

\begin{examp}
The opposite of the nonempty program problem is the {\bf empty program problem}. This tells if the domain of (the number of) a given Turing machine is empty. The empty program problem is undecidable because if it was decidable, then we would be able to compose with the $NOT\colon Bool \lra Bool$ to get a decider for the nonempty program problem
\be \xymatrix{ Nat \ar[rr]^{Empty}\ar@/^.3in/[rrrr]^{Nonempty}&& Bool\ar[rr]^{NOT} && Bool. }\ee
Since we know that nonempty is not computable, we know that empty is not computable.  
\end{examp}

\begin{examp}
The {\bf equivalent program decision problem} asks if two given (numbers of) Turing machines describe the same function. That is, if Turing machine $T$ and $T'$ always give the same output no matter what the input. There is a morphism in $\Func$ called $Equiv\colon Nat \times Nat \lra Bool$ which is defined as follows
\be
  Equiv(y,y') = \left\{
      \begin{array}{ll}
        1 &:  \mbox{if Turing machine $y$ describes the same function as Turing machine $y'$} \\
        0 &:  \mbox{if Turing machine $y$ does not describe the same function as Turing machine $y'$. } \\
      \end{array}
              \right.
\ee
In order to show that $Equiv$ is not in $\CompFunc$ we show that we can reduce $Nonempty$ to $Equiv$ as follows: 
\be \xymatrix{Nat\ar[rr]^h \ar[ddr]_{Empty}&& Nat \times Nat \ar[ddl]^{Equiv}\\ \\& Bool.}\ee
Let $y_0$ be the number of a silly Turing machine that simply goes into an infinite loop for any input. Nothing is ever accepted or output. This machine clearly has a empty domain, i.e, $Empty(y_0)=1$. Now we shall use this Turing machine to describe a reduction from $Empty$ to $Equiv$. $h\colon Nat \lra Nat \times Nat$ is defined for Turing machine $y$ as $h(y)= (y, y_0)$. Notice that $Equiv(y,y_0)=1$ iff Turing machine $y$ performs the same function as the silly Turing machine $y_0$ iff  $Empty(y)=1$. But since we know that $Empty$ is not computable, we know that $Equiv$ is not computable.   
\end{examp}

\begin{examp}
The {\bf printing 42 problem} asks if a given (number of) a Turing machine has some input for which 42 is an output. There is a morphism in $\Func$ called $Print\colon Nat \lra Bool$ which is defined as follows
\be
  Print(y) = \left\{
      \begin{array}{ll}
        1 &:  \mbox{if there exists an input to Turing machine $y$ that outputs 42.} \\
        0 &:  \mbox{if there does not exist an input to Turing machine $y$ that outputs 42. } \\
      \end{array}
              \right.
\ee
Obviously the number 42 is not important to the problem. It is simply the answer to the ultimate question of life, the universe, and everything.  

We show that the halting problem reduces to the printing 42 problem as in 
\be \xymatrix{Nat\times Nat\ar[rr]^{h} \ar[ddr]_{Halt}&& Nat \ar[ddl]^{Print}\\ \\& Bool.}\ee
The total computable function $h$ is defined as follows for Turing machine $y$ and for input $x$. $h(x,y)=y'$ where $y'$ is the number of the Turing machine that performs the following task:

Turing machine $y'$: on input $w$
\begin{enumerate}
\item If $w\neq x$ reject. Stop.
\item If $w = x$ execute Turing machine $y$ on input $x$. If Turing machine $y$ accepts $x$, print ``42'' and accept. Stop.
\end{enumerate}

Notice that Turing machine $y'$ depends on $x$ and $y$. Also notice that function $h$ is easily seen to be totally computable. That means that a program can easily compute $y'$ if it is given $x$ and $y$. Now consider Turing machine $y'$. $Print(y')=1$ iff $Print(h(x,y))=1$ iff Turing machine $y$ accepts $x$ iff $Halt(x,y)=1$. But we already know that it is impossible to solve the halting problem. So it must be impossible to solve the printing problem.  
\end{examp}

There are many other decision problems that can be shown to be undecidable. In fact we will show that a computer cannot deal with the vast majority of properties of Turing machines. What type of properties are we talking about? First we are dealing with nontrivial properties. By this we we mean that there exists Turing machines that have the property and Turing machines that do not have the property. It is very easy to decide trivial properties (just always answer yes or always answer no.) We are also interested in semantic properties. By this we mean we are interested in properties about the function that the Turing machine produces. In other words, if Turing machine $y$ produces the same function as Turing machine $y'$ then both $y$ and $y'$ would both have a semantic property or both not have a semantic property. In contrast to a semantic property, a syntactical property about Turing machines is very easy to decide. For example, it is every to decide if a Turing machine has 100 rules or more. Or if a Turing machine uses less than 37 states. A Turing machine can be written to answer such questions.
 
\begin{theo}{\bf (Rice's theorem)} Any nontrivial, semantic property of Turing machines is undecidable. 
\end{theo}
\begin{proof}
Let $P$ be a nontrivial, semantic property. There will be a morphism is $\Func$ that decides property $P$, i.e., $P-decider \colon Nat \lra Bool$. $P-decider(y)=1$  iff Turing machine $y$ has property $y$. We show that the halting problem is reducible to the problem of deciding the $P$ property with the map $h_P$ 
\be \xymatrix{Nat\times Nat\ar[rr]^{h_P} \ar[ddr]_{Halt}&& Nat \ar[ddl]^{P-decider}\\ \\& Bool.}\ee
Let us say that $y_0$ is the number of the silly Turing machine that always rejects every input. Either Turing machine $y_0$ has property $P$ or does not have property $P$. Assume that it does not (the proof can easily be modified if Turing machine does have property $P$). Since $P$ is nontrivial there exists a Turing machine $y_1$ that does have property $P$. So we have $P-decider(y_0)=0$ and $P-decider(y_1)=1$.
We define $h_P(x,y)=y'$ where $y'$ is the number of the following Turing machine: 

Turing machine $y'$: on input $w$
\begin{enumerate}
\item Simulate Turing machine $y$ on input $x$. 
\begin{enumerate}
\item If it halts and rejects, then reject. Stop. 
\item If it accepts go to step 2.
\end{enumerate}
\item Simulate Turing machine $y_1$ on $w$. 
\end{enumerate}

Notice that if Turing machine $y$ on input $x$ rejects then Turing machine $y'$ will always reject and will be equivalent to Turing machine $y_0$.  If Turing machine $y$ on input $x$ goes into an infinite loop, then $w$ will not be accepted just like Turing machine $y_0$. In contrast, if Turing machine $y$ on input $x$ accepts, then Turing machine $y'$ will act just like Turing machine $y_1$.

Let us analyze Turing machine $y'$. $P-decider(y')=1$ iff $P-decider(h_P(x,y))=1$ iff 	Turing machine $h_P(x,y)=y'$ acts like Turing machine $y_1$ 
iff $Halt(x,y)=1$. But we know that the halting problem cannot be solved. We conclude that $P-decider$ is not computable. 
\end{proof}

Here is just a small sample of the nontrivial semantical properties that Rice's theorem shows are not decidable:
\begin{itemize}
\item Tell if a Turing machine has a finite domain.
\item Tell if a Turing machine has an infinite domain.
\item Tell if a Turing machine accepts a particular input.
\item Tell if a Turing machine accepts all inputs. 
\end{itemize} 

\vspace{.5in}
G{\"o}del's Incompleteness Theorem is one of the most important theorems of 20th century mathematics.  A version of the theorem is a simple consequence of the undecidability of the halting problem. It would be criminal to be so close to it and not state and prove it. 

First some preliminaries. We say a logical system is {\bf complete} if every statement that is true (theorem) has a proof within the system. In contrast a logical system is {\bf incomplete} if there exists a statement that is true for which there is no proof within the system.   
\begin{theo}{\bf (G{\"o}del's Incompleteness Theorem.)} For any consistent logical system that is powerful enough to deal with basic arithmetic, there are statements that are true but unprovable. That is, the logical system is incomplete. 
\end{theo}
\begin{proof}

We will not be going through all the details. However much of the proof has been set up already when we discussed the functor $L$ from the category of Turing machines to the category of logical formulas. Remember that for a Turing machine $T$ and an input $w$ there is a logical formulas $L(T)[w]$ that describes a potential computation of Turing machine $T$ with input $w$. We can use $L(T)[w]$  to formulate a logical formula $HALT(T,w,t)$ which is true exactly when Turing machine $T$ on input $w$ halts in time $t$ or less. This is a logical formula that is either true or false. Whether or not a computation halts then depends on whether or not the logical formula $\exists t HALT(T,w,t)$ is true or not. 

 Because we are dealing with an exact logical system where the axioms are clear and the method of proving theorems are exactly stated, it is possible for a computer to tell when a string is a formal proof of the logical system. This comes from the fact that if the logical system is able to perform basic arithmatic, statements can be encoded as numbers and dealt with. So it is possible (though extremely inefficient) to produce all strings in lexicographical order and for each string have a computer tell if the string is a formal proof of a statement or the negation of a statement. This amounts to saying that there is a computable function $SuperEval \colon String \lra Bool$ that evaluates a logical formula $\phi$ and tells if it or its negation is provably true. It is defined as 
\be
  SuperEval(\phi) = \left\{
      \begin{array}{ll}
        1 &:  \mbox{if there exists a proof that $\phi$ is true} \\
        0 & : \mbox{if there exists a proof that $\neg \phi$ is true. } \\
      \end{array}
              \right.
\ee
The main question is if this computable function is total or not. In a complete logical system, $SuperEval$ is total. In contrast, in an incomplete logical system, there exists statements $\phi$ such that neither $\phi$ nor $\neg \phi$ have proofs and hence $SuperEval$ is not total.  

We will use a reduction from $Halt$ to show that $SuperEval$ is not total and the system is incomplete. There is a total computable function $h \colon Nat \times Nat \lra String$ that makes the following triangle commute 
\be \xymatrix{Nat\times Nat\ar[rr]^h \ar[ddr]_{Halt}&& String \ar[ddl]^{SuperEval}\\ \\& Bool.}\ee
$h$ is defined as 
\be h(x,y)=\exists t HALT(\ulcorner y \urcorner, \ulcorner x \urcorner, t) \ee
where $\ulcorner y \urcorner$  is the Turing machine described by the number $y$ and and $\ulcorner x \urcorner$ is the input string described by the number $x$. Let us examine this function carefully. $ h(x,y)$ is true iff $\exists t HALT(\ulcorner y \urcorner, \ulcorner x \urcorner, t)$ is true. If we were able to prove that $\exists t HALT(\ulcorner y \urcorner, \ulcorner x \urcorner, t)$ or its negation, then we would have a way of deciding the halting problem. We know that is not possible. So it must be the case that there is neither a proof of $\exists t HALT(\ulcorner y \urcorner, \ulcorner x \urcorner, t) $ nor a proof of $\neg \exists t HALT(\ulcorner y \urcorner, \ulcorner x \urcorner, t) $ in the logical system.  For some $y$ and $x$, 
\be SuperEval( \exists t HALT(\ulcorner y \urcorner, \ulcorner x \urcorner, t)) \ee
is undefined. This shows that mathematical  truth is more than the notion of proof. 
\end{proof}
%%%%%%%%%%%%%%%%%%%%%%%%%%

\subsection{Classifying undecidable  problems.}
%%%%%%%%%%%%%%%%%%%%%%%%%%%%%

What is beyond $\CompFunc$? We have shown that there are morphisms that are in $\Func$ and not in $\CompFunc$. While we have given a few examples of such functions, it is important to realize that the vast majority of morphisms in $\Func$ are not in $\CompFunc$. This can easily be seen with a little counting argument. Since every function in $\CompFunc$ can be mimicked by at least one Turing machine and there are only a countable infinite number of Turing machines, we see that there are only a countably infinite number of morphisms in $\CompFunc$. In contrast, there is an uncountably infinite number of morphisms in $\Func$. While we tend to think mostly of what is in $\CompFunc$, there is vastly more morphisms in $\Func$ that are not in $\CompFunc$. 

Is there a way to characterize and classify the morphisms in $\Func$? This was a question that, again, goes back to Alan Turing and he gave an ingenious answer. 

Let $f\colon Seq \lra Seq'$ be any function in $\Func$ (One should think of $f$ as {\it not} being in $\CompFunc$.)
An $f$ {\bf oracle Turing machine} is a Turing machine  that can ``magically'' use $f$ in its computation. In detail, this Turing machine has an extra ``query tape'' and an extra ``query state.'' While the Turing machine is executing 
it can place some information, $x$, on the query tape. Once the question is in place, the Turing machine can go to the special query state. At that time click the oracle will magically  erase $x$ from the query tape and put $f(x)$ on the tape. The computation then continues on its merry way with this new piece of information. 

For any morphism $f$ in $\Func$ we formulate the category of $f$ oracle Turing machines which we denote $\Turing[f]$. (The notation should remind a mathematician of taking a ring and adding in an extra variable to get a larger ring.) The objects of  $\Turing[f]$ are the natural numbers. The morphisms are $f$ oracle Turing machines. A regular Turing machine can be thought of as an oracle Turing machine where the query tape and the query state is never used. This means that there is an inclusion functor $\Turing \longhookrightarrow \Turing[f]$. We can also discuss what functions can be mimicked by a Turing machine that has access to the $f$ oracle. This gives us the category $\CompFunc[f]$. If a function does not use the oracle, then it is in $\CompFunc$. Hence there is an inclusion $\CompFunc \longhookrightarrow \CompFunc[f]$. Notice that if $f$ is computable by itself then $\CompFunc[f]$ is the same thing as $\CompFunc$ because, rather than using the oracle, we can just put in a subroutine that computes the function. Every morphism in $\CompFunc[f]$ is still a morphism in $\Func$. We can summarize all these categories with this diagram.  
\be \xymatrix{
&&\Turing\ar@{->>}[dd]_D \ar@{^{(}->}[r]& \Turing[f]\ar@{->>}[dd]\\
\\
&&\CompFunc  \ar@{^{(}->}[r]& \CompFunc[f] \ar@{^{(}->}[r]&\Func.\\
}\ee

Rather than just taking any arbitrary non-computable $f$ for an oracle, let us take $Halt\colon Nat\times Nat \lra Bool$. This will result in the category $\CompFunc[Halt]$ which consists of all the functions that are computable if a computer has access to the $Halt$ function. This is more than $\CompFunc$ but not all of $\Func$. We can ask whether or not a Turing machine with a $Halt$ oracle will halt. These Turing machines can also be enumerated and we can make a new halt function \be \widehat{Halt}\colon Nat \times Nat \lra Bool\ee which is defined as 
\be
 \widehat{Halt}(x, y) = \left\{
      \begin{array}{ll}
        1 &:  \mbox{if $Halt$ oracle Turing machine $y$ on input $x$ halts.} \\
        0 &: \mbox{if $Halt$ oracle Turing machine $y$ on input $x$ does not halt. } \\
      \end{array}
              \right.
\ee

It is not hard to show that the $\widehat{Halt}$ is not computable even with the $Halt$ oracle. That is, $\widehat{Halt}$ is not in $\TotCompFunc[Halt]$. We can use $\widehat{Halt}$ as a new oracle and construct $\CompFunc[\widehat{Halt}]$. This process of going from one category of functions to a larger category of functions is called the {\bf jump operation}. We can continue this process again and again. We have the following infinite sequence of categories:
\be\CompFunc \longhookrightarrow \CompFunc[Halt] \longhookrightarrow \CompFunc[\widehat{Halt}] \ee
\be  \longhookrightarrow \CompFunc[\widehat{\widehat{Halt}}] \longhookrightarrow \cdots \longhookrightarrow \Func.
\ee
This gives us a whole lattice of categories where $\CompFunc$ is the bottom and $\Func$ is the top. This is a classification of the uncomputable functions in $\Func$.  We know a lot about what we cannot compute. This beautiful structure is extensively studied in books like \cite{soare} and \cite{rogers}.  

%
% Complexity Theory
%
\section{Complexity Theory}
While computability theory deals with what can and cannot be computed, complexity theory deals with what can and cannot be computed {\it efficiently}.
Here we do not ask what functions are computable. Rather what computable functions can be computed with a reasonable amount of resources. We also classify different types of computable functions by their different levels of efficiency or complexity. 

Historically, complexity theory only deals with total computable functions. Our entire discussion will only deal with the following functor from The Big Picture.  
\be \xymatrix{
\TotTuring\ar@{->>}[dd]_D\\ 
\\
\TotCompFunc 
}\ee

\subsection{Measuring Complexity}

When we discuss a function using an efficient amount of resources we usually mean number of steps to compute the function.  This corresponds to the amount of time it takes for the computation to complete. The more steps needed to complete the computation, the more time will be needed. Researchers have also been interested in how much space or other resources are required to compute the function.  Let us focus on the amount of time a computation needs. Usually the amount of time needed depends on two things: (i) the size of the input. One expects that a large input would demand a lot of time and a small input can be done rather quickly. And (ii) the state of the input. For example, if one is interested in sorting data, then usually, data that is already almost sorted, does not require a lot of computing time to get the data totally in order. In contrast, if the data is totally disordered, more time is needed. We will be interested in the worst-case scenario, that is the worst possible state of the data.  

In order to formalize the notion that the amount of operations and time needed is dependent on the size of the input, we associate to every total  Turing machine a function from the natural numbers, $\n,$ to the set of non-negative real numbers $\r^\ast$. The function $f\colon \n \lra \r^\ast$ will tell how many operations are needed for a computation of the Turing machine. The $\n$ corresponds to possible sizes of the input and the $\r^\ast$ corresponds to the possible number of operations needed. (While only whole numbers are used to describe how many operations are required, since some functions will use operations like $log$ we employ the codomain  $\r^\ast$.)  If $n$ is the size of the input, then $f(n)$ is the amount of operations needed in the worst-case scenario. $f(n)$ is the maximum amount of resources needed for inputs of size $n$. 

All this is for one Turing machine that implements the function. But there might be many Turing machines that implement the same function. We are going to need to consider the best-case scenario, i.e., the most efficient Turing machine that implements a certain function. So to every total computable function we will associate a function $f\colon \n \lra \r^\ast$  which is calculated by looking at the best possible Turing machine that solves that problem and looking at the worst possible data that can be input to that machine. 

The set of all functions $\n \lra \r^\ast$ form an ordered monoid $Hom_\Set(\n,\r^\ast)$. The monoid operation, $+$, is inherited from $\r^\ast$. The unit is the function that always outputs zero. The order is also inherited from $\r^\ast$. Essentially $f \leq g$ if and only if $f(n) \leq g(n)$ for all $n \in \n$. We shall think of $Hom_\Set(\n,\r^\ast)$ as a one-object category.

\begin{tech} 
When complexity theorists compare two Turing machines they are not really looking at the ordered monoid $Hom_\Set(\n,\r^\ast)$. Rather, they want to consider two functions to be the same if they only differ by a small amount. They also want to ignore what happens when the input sizes are small. Researchers deal with this by working with a quotient ordered monoid $Hom_{\Set}(\n, \r^\ast)/\sim$ defined as follows: 
$f\colon \n \lra \r^\ast$ is considered the same as $g\colon \n \lra \r^\ast$, i.e.,  
$f \sim g$ if and only if \be0 < \lim_{n\to\infty}\f{f(n)}{g(n)}< \infty.\ee It is not hard to see that this relation is an equivalence relation. In fact, it is a congruence. 

Rather then using the usual equivalence class notation, i.e.,  $f\in [g]$, complexity theorists use the notation $f \in \Theta(g)$ or $f =\Theta(g)$. Also, if $f \leq g$ in the quotient ordered monoid, we write $f=O(g)$. 

Although this quotient plays a prominent role in complexity theory, we do not lose the spirit of the subject by ignoring it.  
\end{tech}

The categorical details of how we come to a functor that measures complexity is not simple. Feel free to skip to the punchline in Diagram \ref{diag:measurecomplexity} on page \pageref{diag:measurecomplexity}.

We need to look at all the Turing machines with all their appropriate input which we construct from the pullback  
\be \xymatrix{\InpTM \ar[dd]\ar[rr]&& \coprod_m\coprod_nHom_{\TotTuring}(m,n)\ar[dd]
\\
\\
\coprod_{m=0}^{\infty} (\Sigma^*)^m \ar[rr] && d(\n).
}\ee
Let us explain this pullback. The lower right entry is the discrete set of natural numbers. The upper right corner is the set (as opposed to the category) of all Turing machines. The right vertical functor takes every Turing machine to the number of string inputs it demands. The lower left is the set of all possible input strings. The $m$-tuples of strings for all $m$ is the set $\coprod_{m=0}^\infty (\Sigma^*)^m$.  The bottom functor takes an $m$-tuple of strings and simply outputs $m$.  The pullback of these two functors is the set of pairs of Turing machines and the inputs for those Turing machines. We call this discrete category $\InpTM$.
 
We must be able to measure the size of the input to a Turing machine. Not only do we need to measure the number of strings, but we also need to know the sum total of the lengths of all the strings. There is a length functor $|~~| \colon \Sigma^* \lra \n$ that takes a string and gives the length of the string. This functor can be extended to an $m$-tuple of strings  $|~~| \colon (\Sigma^*)^m \lra \n$ in the obvious way: if $x_1, x_2, x_3, \ldots , x_m$ is an $m$-tuple of strings, then
\be | x_1, x_2, x_3, \ldots , x_m| = |x_1|+ |x_2|+ |x_3|+ \cdots + |x_m|.\ee
This can be further extended to $m$-tuples for any $m$. This gives us the functor $|~~| \colon \coprod_m (\Sigma^*)^m \lra \n$.

Given a Turing machine and its input there are several resources we can measure. We can measure the number of computational steps that this Turing machine with that input will demand to complete the computation. We call this $Time$. We can measure the number of boxes needed in the work tape to complete the computation. We call this resource $Space$. There are still other resources studied. 
\be \xymatrix{\InpTM \ar@<.2 in>[rr]^{Time} \ar@<.05 in>[rr]_{Space} \ar@<-.3 in>[rr]_{others}^\vdots & &\r^*}\ee

Let us examine important subsets of $\InpTM$. For Turing machine $T$ that demands $m$ inputs there is a subset $\la T, (\Sigma^*)^m\ra$ of $\InpTM$ that consists of all the possible inputs to that Turing machine. Each of the measures of resources restrict to this set.  From such a set, and for any measure of resource, say $Time$, there is a function to $\r^*$ and a length function to $\n$, i.e.,
\be \xymatrix{\n &&\la T, (\Sigma^*)^m\ra\ar[rr]^{Time}\ar[ll]_{|~~~|}&&\r^*.}\ee 
%There is an adjunction 
%\be \xymatrix{(\r^*)^{\la T, (\Sigma^*)^m\ra }\ar@<.1 in>[rr]^{max_{T}}_\top && \r^{*\n}\ar@<.1 in>[ll]^{|~~~|\circ -}} \ee
%This adjunction says that for a resource $resource\colon \la T, (\Sigma^*)^m\ra \lra \r^*$ in $(\r^*)^{\la T, (\Sigma^*)^m\ra }$ and a function corresponding to an algorithm $g\colon \n \lra \r^*$ in $ \r^{*\n}$ we have the following isomorphism of hom-sets
%\be(\r^*)^{\la T, (\Sigma^*)^m\ra }(|~~~|\circ g ,f) \cong (\r^{*})^{\n}(g, max_{T,Time}(f) )  \ee 
These two functions are used to define a function 
\be\max_{T}(Time) \colon \n \lra \r^*.\ee
This is the function that gives the resources with regard to all input $w$. In symbols this is 
\be \max_{T}(Time) = \lim_{\substack{g\colon \n \lra \r^*\\ w\in (\Sigma^*)^m\\ g(|w|)=Time(T,w) }} g\colon \n \lra \r^*\ee

\begin{tech} For those who know the language of Kan extension, this colimit is nothing more than the left Kan extension
\be \xymatrix{
\n \ar[rr]^{\max_{T}(Time)}&&\r^*
\\
\\
& \la T, (\Sigma^*)^m\ra \ar[luu]^{|~~|}\ar[ruu]_{Time}
}\ee
\end{tech}

We are still not done. We have found the function that describes the amount of time needed for a particular Turing machine $T$. It remains to find the function that describes resources needed by searching through all the Turing machines that implement some computable function.  Let $f\colon Seq \lra Seq'$ be some total computable function, i.e.,  in $\TotCompFunc$. Consider the preimage set $D^{-1}(f)$ of all Turing machines that implement $f$.  We define the function that measures complexity of computable functions as follows:
\be \mu_{D,Time}(f)\  =  \min_{T \in D^{-1}(f)} \max_{T}(Time)\ee
This gives us the desired functor: 
\be \label{diag:measurecomplexity}
\xymatrix{\TotCompFunc \ar[rr]^{\mu_{D,Time}}&&
Hom_\Set(\n,\r^*) }
\ee
Notice the functor $D$ is used in the notation. If we use the resource $Space$, we will get the functor that measures space $\mu_{D, Space}$.

The functors  $\mu_{D, Time}$ and $\mu_{D, Space}$  are not just set functions. Sequential processes in $\TotCompFunc$ go to the sum of functions in $Hom_\Set(\n,\r^*)$. This means that if two functions are performed one after the other, then the amount of resources needed will be added. 

In order to classify the computable functions, we look at various submonoids of $Hom_\Set(\n,\r^*)$.  For example, we will look at the submonoids $\mathbf{Poly}$ of all polynomial functions, $\mathbf{Const}$ of all constant functions, $\mathbf{Exp}$ of all exponential functions, $\mathbf{Log}$ of all logarithm functions, etc. These monoids are included in each other as
\be \mathbf{Const} \longhookrightarrow \mathbf{Log} \longhookrightarrow  \mathbf{Poly} \longhookrightarrow \mathbf{Exp} \longhookrightarrow  Hom_\Set(\n,\r^*)\ee

 For every such submonoid, we can take the following pullback and get those computable functions whose complexity is within that submonoids. For example, for polynomials, we have the pullback:
\be \label{diag:pbcomplexity}
\xymatrix{
\Poly_{D, Time}\ar@{^{(}->}[rr]\ar[dd]&&\TotCompFunc \ar[dd]^{\mu_{D,Time}}\\
\\
\mathbf{Poly}\ar@{^{(}->}[rr] &&Hom_\Set(\n,\r^*) }
\ee
The category $\Poly_{D, Time}$ is the collection or {\bf complexity class} of all computable functions that can be computed in a  polynomial amount of time. Another complexity class is $\Exp_{D, Space}$, the collection of all computable functions that can be computed using an exponential amount of space. The notation is self evident.

At this time, we would like to bring in a more advanced notion of a Turing machine. Our entire discussion has been about {\bf deterministic Turing machines}. These are machines that, at every single time click, do exactly one operation. There are souped-up Turing machines called {\bf nondeterministic Turing machines} that at every time click might do one of a set of possible operations. 
A nondeterministic Turing machine could be in state $q_{32}$ and see various symbols on its tapes, it has the options of doing one of a set of operations on the tapes. In analogy with Equation \ref{diag:TurDef}, we can write this as follows:
 \begin{align} 
 \delta(q_{32}, x_1,x_2, \ldots,x_n)  = &\{(q_{51}, y_1,y_2, \ldots,y_n, L, R, R, \ldots, L)\\
					& (q_{13}, y'_1,x_2, \ldots,y'_n, R, R, \ldots, L) ,\\
					& (q_{51}, y^2_1,x^2_2, \ldots,x^2_n, L, R, \ldots R). \}      
\end{align}
In this example the Turing machine has three different options of operations to perform. A computation begins when input is placed on the input tape. At every time click the Turing machine can choose one of the possible options. We say that a computation occurs when there is a sequence of choices that leads to an accepting state. The first such sequence of choices gives us the computation.

There is a category of nondeterministic  Turing machines, $\NTotTuring$,  The objects are the same natural numbers as with $\TotTuring$ and the set of morphisms from $m$ to $n$ is the set of nondeterministic Turing machines with $m$ input tapes and $n$ output tapes. Analogous to the functor $D\colon \TotTuring\lra \TotCompFunc$ there is a functor $N\colon \NTotTuring\lra \TotCompFunc$ that takes every Turing machine to the function it computes. Every deterministic Turing machine can be thought of as a special type of nondeterministic Turing machine where the set of options is a singleton set. There is an obvious inclusion functor from $\TotTuring$ to $\NTotTuring$. There exists a functor $F \colon \NTotTuring \lra \TotTuring$ that takes every nondeterministic Turing machine to a determinstic Turing machine that performs the same computable function. The deterministic Turing machine works by trying every possible path of the nondeterministic Turing machine. We summarize with the following:  
\be \xymatrix{
\TotTuring\ar@{->>}[ddr]_D\ar@{^{(}->}[rr]^{Inc}&&\NTotTuring\ar@{->>}[ddl]^N\ar@/_.3in/[ll]_F\\ 
\\
&\TotCompFunc 
}\ee
It should be noted that while $F \circ Inc = Id_\TotTuring$, it is not necessarily true that $Inc \circ F = Id_\NTotTuring$. However, it is true that $N \circ Inc \circ F= N$. This means that for every nondeterministic Turing machine there is a deterministic Turing machine that performs the same computable function. 
That is, every nondeterministic Turing machine there is an equivalent deterministic Turing machine.

We discuss measuring resources by replacing the functor $D$ with the functor $N$ in the definition of Diagram  \ref{diag:measurecomplexity} to get  
\be 
\xymatrix{\TotCompFunc \ar[rr]^{\mu_{N,Time}}&&
Hom_\Set(\n,\r^*) }.
\ee
Pullbacks like \ref{diag:pbcomplexity} can be used to form categories like $\Poly_{N,Time}$ or $\Exp_{N,Space}$ etc. Since every computable function that can be performed by a deterministic Turing machine in polynomial time can also be performed by nondeterministic Turing machine in polynomial time, there is an induced  inclusion functor from $\Poly_{D,Time}$ into the category of $\Poly_{N,Time}$ as in 
\be 
\xymatrix{\Poly_{D, Time}\ar@{^{(}->}[drr]^\eu \ar@{^{(}->}[rrrrrr]\ar[rrddd]&& && && \TotCompFunc \ar[dddll]^{\mu_{D,Time}}\ar[dll]_=\\
&&\Poly_{N, Time}\ar@{^{(}->}[rr]\ar[dd]&&\TotCompFunc \ar[dd]^{\mu_{N,Time}}\\
\\
&&\mathbf{Poly}\ar@{^{(}->}[rr] &&Hom_\Set(\n,\r^*) 
}
\ee

Similarly other submonoides of $Hom_\Set(\n,\r^\ast)$ induce other inclusions as in 

\be 
\xymatrix{
\Log_{D,Time} \ar@{^{(}->}[ddd]\ar@{^{(}->}[drr]^\eu\ar[rrrrd]
\\
&&\Poly_{D, Time}\ar@{^{(}->}[rr]\ar[dd]&&\TotCompFunc \ar[dd]^{\mu_{D,Time}}\\
\\
\mathbf{Log}\ar@{^{(}->}[rr]&&\mathbf{Poly}\ar@{^{(}->}[rr] &&Hom_\Set(\n,\r^*) }
\ee

Diagram \ref{diag:complexityclasses} shows how all these various subcategories are related. 

\begin{figure}[h]
\be \xymatrix{
\Exp_{D,Time}\ar@{^{(}->}[rr]&&\Exp_{N,Time}\\
\\
\Poly_{D,Time}\ar@{^{(}->}[uu]\ar@{^{(}->}[rr]&&\Poly_{N,Time}\ar@{^{(}->}[uu]\\
\\
\Log_{D,Time}\ar@{^{(}->}[uu]\ar@{^{(}->}[rr]&&\Log_{N,Time}\ar@{^{(}->}[uu]\\
\\
\Const_{D,Time}\ar@{^{(}->}[uu]\ar@{^{(}->}[rr]&&\Const_{N,Time}\ar@{^{(}->}[uu] 
}\ee
\caption{Some subcategories (complexity classes) of $\TotCompFunc.$}\label{diag:complexityclasses}
\end{figure}

We only dealt with deterministic and nondeterministic Turing machines. 
There are many other types of Turing machines that we will not discuss. There are probabilistic Turing machines, quantum Turing machines, alternating Turing machine, etc. Each with its own set of rules and with its own complexity classes. The relationship between all these complexity classes are a major topic within complexity theory. 

\subsection{Decision problems} 

As in  computability theory, there is a special interest in decision problems. In this context, decision problems are total computable functions whose codomain is $Bool$.
As we saw in computability theory, there is a way of comparing decision problems. In complexity theory, we are interested in special types of reductions from one decision problem to another.
\begin{defi}
Let $f\colon Seq \lra Bool$ and $g\colon Seq' \lra Bool$ be two decision functions in $\TotCompFunc$. We say that $f$ is {\bf polynomial reducible} to $g$ if there is a $h\colon Seq \lra Seq'$ in $\Poly_{D,Time}$ such that 
\be \label{diag:polyreduc}\xymatrix{Seq\ar[rr]^h\ar[ddr]_f&& Seq' \ar[ddl]^g\\ \\& Bool.}\ee
We write this as $f \leq_p g$. If we further have that $g \leq_p f$ then we write $f\equiv_p g$ and say they are in the same {\bf complexity class}. 
\end{defi}

We form the category of decision problems and polynomial reductions. Consider the functors 
\be \xymatrix{\Poly_{D,Time}\ar@{^{(}->}[rr]^{Inc}&&\TotCompFunc && \one \ar[ll]_{Const_{Bool}}}\ee 
where the left functor is an inclusion functor and $Const_{Bool}$ takes the single object in $\one$ to the type $Bool$. Now consider the comma category $(Inc,Const_{Bool})$. The objects of this category are computable decision problems and the morphisms are polynomial reductions from one decision problem to another. 

There are two subcategories of $\TotCompFunc$ that are of interest: $\Poly_{D,Time}$  and $\Poly_{N,Time}$. These are all deterministic polynomial computable functions and all nondeterminisitic polynomial functions, respectively. They sit in the diagram
\be\label{diag:commacats} \xymatrix{&&\TotCompFunc 
\\
\\
\Poly_{D,Time}\ar@{^{(}->}[rr]^{DInc} \ar@{^{(}->}[rruu]^{Inc}\ar@{^{(}->}[rrdd]_{Id}&& \Poly_{N,Time}\ar@{^{(}->}[uu]&&\one \ar[ll]_{Const_{Bool}} \ar[uull]_{Const_{Bool}}\ar[ddll]^{Const_{Bool}}
\\
\\
&&\Poly_{D,Time}\ar@{^{(}->}[uu]
 }\ee  
The comma category $(DIinc, Const_{Bool})$ which consists of nondeterministic polynomial decision problems and (deterministic) polynomial reductions is called the complexity class $\NP$. The comma category $(Id, Const_{Bool})$ which consists of deterministic polynomial decision problems and (deterministic) polynomial reductions is called the complexity class $\PT$. The inclusion $\Poly_{D,Time} \longhookrightarrow \Poly_{N,Time}$ induces the inclusion $\PT \longhookrightarrow \NP$.

The most prominent open problem in theoretical computer science is the $\PT ~=?~ \NP$ question. While it is known that $\PT$ is a subcategory of $\NP$, it remains an open question to tell if these categories are really the same category. In other words, is there a morphism in $\NP$ that is not in $\PT$ or is every morphism in $\NP$ also in $\PT$. Alas, this question will not be answered in this mini-course. 

The notion of polynomial reduction is very important. In Diagram \ref{diag:polyreduc}, since $h$ is in $\Poly_{D,Time}$ we have that if $g$ is also in $\Poly_{D,Time}$ then by composition, so is $f$. That is, if $g$ is in $\PT$, then $f$ is in $\PT$. The contrapositive of this statement is more interesting: If $f$ is not in $\Poly_{D,Time}$, then neither is $g$. That is, if $f$ is not in $\PT$, then neither is $g$ in $\PT$. 

A  terminal object $t$ in a category is an object such that for any object $a$ there is exactly one morphism $a \lra t$. Define a {\bf weak terminal object} $w$ in a category to be an object such that for every object $a$ in the category there is {\it at least one} morphism $a\lra w$. Consider the full subcategory of $\NP$ of all weak terminal objects. The weak terminal objects are called {\bf NP-Complete problems} and the full subcategory of all of them is $\NPComplete$. These are the nondeterministic polynomial decision problems such that every nondeterministic polynomial decision problem  polynomial reduces to it. There might be more than one reduction. We have the inclusion of categories $ \NPComplete \longhookrightarrow \NP.$  
NP-Complete problems are very important in complexity theory. They are central to the $\PT= \NP$ question. If one shows that any particular NP-Complete problem can be solved or decided by a polynomial Turing machine then all the morphisms in $\NP$ can be shown to be in $\PT$ and $\PT=\NP$. In contrast, if we can find one morphism in $\NP$ that does not have a polynomial Turing machine, then we can show that $\PT\neq \NP$. 

Given a weak terminal object $w$, any map $h\colon w \lra w'$ insures that the object $w'$ is also a weak terminal object. In terms of NP-Complete problems, this means that if $f$ is an NP-Complete problem and there exists a polynomial reduction from $f$ to $g$, then $g$ is also an NP-Complete problem. So to find a cadre of NP-Complete problems, we have to find a single one first. Logic gives us this example.     

The {\bf Satisfiabilty problem} accepts a Boolean formula and asks if there is a way to assign values to the variables that make the formula true. Can the logical formula be satisfied? The usual way this is done is to fill out a truth table of the formula and see if there is any ``true'' in the final column. This describes a computable morphism $SAT \colon String \lra Bool$ in $\TotCompFunc$.

We would like to show that $SAT$ is a NP-Complete problem. Let us emphasize what this means. If $SAT$ is NP-complete then {\it every} NP problem reduces to it. Over the past several decades, researchers have described thousands of NP problems. There are still thousands more to be described in the future. How are we to show that everyone of these problems can be reduced to $SAT$?  

\begin{theo}{\bf (The Cook-Levin Theorem.)} 
$SAT\colon String \lra Bool$ is a weak terminal object in $\NP$. That is, $SAT$ is an NP-Complete problem.
\end{theo}
\begin{proof}
We shall only give the bare outline of the proof.
We have to show that for any $g\colon Seq \lra Bool$ in $\NP$ there is a polynomial reduction $h_g\colon Seq \lra String$ such 
\be \xymatrix{Seq\ar[rr]^{h_g}\ar[ddr]_{g}&& String \ar[ddl]^{SAT}\\ \\& Bool}\ee commutes. 
The one thing we know about every problem in $\NP$ is that, by definition, there is a computer that can execute the function.  
Consider the following sequence of functors
\be \xymatrix{\NP\ar@{^{(}->}[rr]^{Inc} && \TotCompFunc \ar[rr]&& \Turing \ar[rr]^L&&\Logic }
\ee
Call the composition of these functors $L'$ for ``logic''. $L(g)$ is a logical formula with variables that describes the workings of $g$. If $x$ is an input to $g$, then $g(x)$ is either true or false. We define $h_g(x)$ to be the logical formula $L'(g)$ with some of the variables set to the input values of the function. 
We will write this as $L'(g)[x]$. The point of the construction is that $g(x)$ is true if and only if $L'(g)[x]$ is true. The hard part of the proof is to show that $h_g$ is polynomial.  One can find the complete proof in Section 2.6 of \cite{garey}, Section 34.3 of \cite{corman}, Section 7.4 of \cite{sipser} 
\end{proof}
%%%%%%%%%%%%%%%%

\subsection{Space Complexity}
Till now we have concentrated on the time resource. Let us give one of the main results about the space resource. When dealing with time complexity, the big open question is the relationship between $\PT$ and $\NP$. In contrast, the analogous question for space complexity is answered. 

First some preliminaries. The resources measured are in a deterministic Turing machine computation is the number of cells on the work tape used. For a nondeterministic Turing machine, we measure the number of cells used in an accepting computation. As we did with time complexity in Diagram \ref{diag:measurecomplexity}, we can formulate the functors 
\be \xymatrix{\TotCompFunc \ar[rr]^{\mu_{D,Space}}&&
Hom_\Set(\n,\r^*) }.\ee
and 
\be \xymatrix{\TotCompFunc \ar[rr]^{\mu_{N,Space}}&&
Hom_\Set(\n,\r^*) }.\ee
Since every cell used on a Turing tape demands a time click, we can show that for every $f\colon Seq \lra Seq'$ in $\TotCompFunc$ we have 
\be \mu_{D, Time}(f) \leq \mu_{D,Space}(f) \mbox{             and              } \mu_{N, Time}(f) \leq \mu_{N,Space}(f)\ee
Using these functors and the submonoid {\bf Poly} we can use pullbacks analogous to Diagram \ref{diag:pbcomplexity}
to form $\Poly_{D, Space}$ and $\Poly_{N,Space}$. These are subcategories of $\TotCompFunc$ that correspond to total computable functions that can be 
computed using a polynomial amount of space deterministicly and nondeterministicly respectively. Using diagrams analogous to Diagram 
\ref{diag:commacats}, we can form the category of decision problems $\PSPACE$ and $\NPSPACE$.  
\begin{theo}{\bf (Savitch's Theorem.)} 
The inclusion function $\PSPACE \longhookrightarrow \NPSPACE$ is actually the identity. That is, $\PSPACE=\NPSPACE$. 
\end{theo}
\begin{proof}
The proof basically shows that every nondeterminstic computation that uses $f(n)$ spaces can be mimicked by a deterministic computation that uses $f(n)^2$ spaces. In particular, if $f(n)$ is a polynomial, then $f(n)^2$ is also a polynomial. We conclude that all the nondeterministic computable decision problems in $\NPSPACE$ are also in $\PSPACE$. The details of the proof can be found in Section 8.1 of \cite{sipser}, Section 7.3 of \cite{pap} and Section 4.2 of \cite{barak}. 
\end{proof}
\section{Kolmogorov Complexity Theory}
In this area of theoretical computer science we measure the informational content of strings. We say the Kolmogorov complexity of string $w$ is the size of the smallest Turing machine that can produce $w$. The idea is that the string is a simple string then a small Turing machine can produce the string. In contrast, if the string is more complicated and has more informational content, then the Turing machines needs to be more complicated. What if the string is so complicated, that there are no small Turing machines that can produce it? 

First some motivating examples. Consider the following three strings:
\begin{enumerate}
\item 00000000000000000000000000000000000000000000000
\item 11011101111101111111011111111111011111111111110
\item 01010010110110101011011101111001100000111111010
\end{enumerate}

All three consists of 0's and 1's and are of length 45. It should be noted that if you 
flipped a coin
45 times the chances of getting any of these three sequences are equal. That is, the chances for
each of the strings occurring is $\f{1}{2^{45}}$. This demonstrates a fault of classical probability theory in
measuring the informational content of a string. Whereas you would not be shocked to see a sequence of coins
produce string 3, the other two strings would be surprising. A better way of measuring the informational content is to look at the shortest programs that describe these strings:

\begin{enumerate}
\item Print 45 0's.
\item Print the first 6 primes.
\item Print `01010010110110101011011101111001100000111111010'.
\end{enumerate}
The shorter the program, the less informational content of the string and the string is ``compressible.'' In contrast, if only a long
program can describe the string, then the string has more content. If the only way to have a computer formulate the string is to literally have the string in the program then the string is  ``incompressible.'' An incompressible string is also called ``random'' because it has no patterns that we can use to print it out.

We should note that Kolmogorov complexity theory is not the only way to measure strings. There is computational complexity (how many steps does it take for the Turing machine to print the string), logical depth \cite{Bennett}, sophistication \cite{Koppel} and others. 

Let us use the categories from The Big Picture. We only need to look at the functor $D \colon \TotTuring \lra \TotCompFunc$. But even this is too complicated. Let us look at the restriction functor $D|$ that we get from the following pullback:
\be \xymatrix{
\TotTuring(1,1)\ar@{->>}[dd]_{D|}\ar@{^{(}->}[rr]&&\TotTuring\ar@{->>}[dd]^D\\ 
\\
\TotCompFunc(Str,Str)\ar@{^{(}->}[rr]&&\TotCompFunc. 
}\ee
That is, $D|$ is the functor (actually it is a set function) from the set of all total Turing machines with one input tape and one output tape to the set of computable functions that accept a string and output a string. There is a size functor $Sz\colon \TotTuring(1,1) \lra d(\n)$ that assigns to every total Turing machine the number of rules in the Turing machine. 

\begin{defi}
Let $x$ and $y$ be strings. Then we define the {\bf relative Kolmogorov complexity} to be the size of the smallest Turing machine that mimics a computable function that for the input $y$, outputs $x$. In symbols.
\be K(x|y)=\min_{\substack{T\in D|^{-1}(f\colon String \lra String) \\ f(y)=x}}Sz(T)\ee
If $y$ is the empty string, then $K(x)=K(x|\empty)$ is the {\bf Kolmogorov complexity} of $x$. This is the size of the smallest Turing machine that starts with an empty tape and outputs $x$.  
\end{defi}
If $K(x)\leq |x|+c$ for come constant $c$ then $x$ is {\bf compressible}  Otherwise $x$ is {\bf incompressible} and {\bf random}.  

We end this short visit into Kolmogorov complexity theory with the main theorem about $K$. 
One might believe that $K$ can be computed and we can find the exact amount of minimal structure each string contains. Wrong.  
\begin{theo}K is not a computable function.
\end{theo}
\begin{proof}
The proof is a proof by contradiction. Assume (wrongly) that $K\colon String \lra Nat$ is a computable function. We will use this function to show a contradiction. If $K$ is computable, then we can use $K$ to compute a computable function $K'\colon Nat \lra String$. $K'$ works as follows:
\begin{enumerate}
\item Accept an integer $n$ as input. 
\item Go through every string $s \in \Sigma^*$ in lexicographical order
\begin{enumerate}
\item Calculate $K(s)$.
\item If $K(s)\leq n$, continue. 
\item If $K(s)> n$, output $S$ and stop. 
\end{enumerate}
\end{enumerate}

For any $n$ this program will output a string with a larger Kolmogorov complexity than $n$.
This program has a size, say $c$. If we “hard-wire” a number $n$ into the program, this would
demand $\log n$ bits and the entire program will be of size $\log n + c$. Hard-wiring means making a computable function $K''\colon \ast \lra Nat \lra String$ where $\ast \lra Nat$ picks out $n$. 
$K''$ will output a string that demands more complexity than $n$. Since we
can find an $n$ such that $n > \log n + c$, $K''$ will produce a string that has higher Kolmogorov complexity then the size of the Turing machine that produced it. This is a contradiction. Our assumption that $K$ is computable is false. 
\end{proof}
This means we can never get a computer to tell us if there is a string has more structure than what we see. 

% Algorithm
\section{Algorithms}
We close this mini-course with a short discussion of the definition of ``algorithm''. Notice that this word has not been mentioned so far. While we freely used the words ``function,'' ``program,'' ``Turing machine,'' we did not use ``algorithm.'' This is because the formal definition of the word is not simple to describe.

There are those who say that an algorithm is exactly the same thing as a program. In fact, on page 5 of the authoritative \cite{corman}, an algorithm is informally defined as ``any well-defined computational procedure that takes some value, or set of values, as input and produces some value, or set of values, as output.'' We are left with asking what is a  ``procedure''? Furthermore, this informal definition seems like it is defining a program not an algorithm.  

The problem with equating algorithms with programs is that is not the way the word ``algorithms'' is used colloquially. 
If there are two programs that are very similar and only have minor differences, we usually do not consider them different algorithms. We say that the algorithm is the same but the programs are different. Here are some examples of differences in programs that we still consider to be the same algorithm:
\begin{itemize}
\item One program uses variable name $x$ for a certain value, while the other program uses variable name $y$ for the same value.
\item One program performs a process $n$ times in a loop, while another program performs the process in a loop $n-1$ times and then does the process one more time outside the loop.  
\item One program performs two unrelated processes (they do not effect each other) in one loop, while a second program performs each of the two unrelated processes in their own separate loop.  
\item One program performs two unrelated processes in one order while a second program performs the unrelated processes in the reverse order.  
\end{itemize}
This list can easily be extended. 

Let us make the case in another way. A teacher describes a certain algorithm to her computer class. She tells her thirty students to go home and implement the algorithm. Assuming that they are all bright and that there is no cheating, thirty {\it different} programs will be handed in the next class. Each program is an implementation of the algorithm. While there are differences among the programs, they are all ``essentially the same.'' All the programs definitely implement the same function. 
This is the way that the word ``algorithm'' is used.

With this is mind, we make the following definition.
\begin{defi} Take the set of all programs. We shall describe an equivalence relation on this set where two programs are equivalent if they are ``essentially the same.'' An {\bf algorithm} is an equivalence class of programs under this relation. Note that all the programs in the same equivalence class perform the same computable function. However there can be two different  equivalence classes that also perform the same function. 
\end{defi} 

Figure \ref{pic:defofanalg} on page \pageref{pic:defofanalg} makes this all clear. There are three levels. The top level is the collection of all programs. The bottom level is the collection of all computable functions. And in-between them is the collection of algorithms. Consider programs $mergsort_a$ and $mergesort_b$. The first program is an implementation of mergesort programed by Alice while the second program is written by Bob. They are both implementations of the algorithm $mergesort$ found in the middle level. There are also programs $quicksort_x$ and $quicksort_y$ that are different implementations of the algorithm $quicksort$. Both the algorithms $mergesort$ and $quicksort$ perform the same computable function $sort$. The big circle above the cone represented by $sort$ contains all the programs that implement the sort function. Above the computable function $find~~max$ there are all the programs that take a list and find the maximum element. Some of those programs are essentially the same and are implementations of the $binarysearch$ algorithm while others implement the $brute~~search$ algorithm. 

\begin{figure}[!ht]
\centering
 \includegraphics[width=15cm, height=12cm]{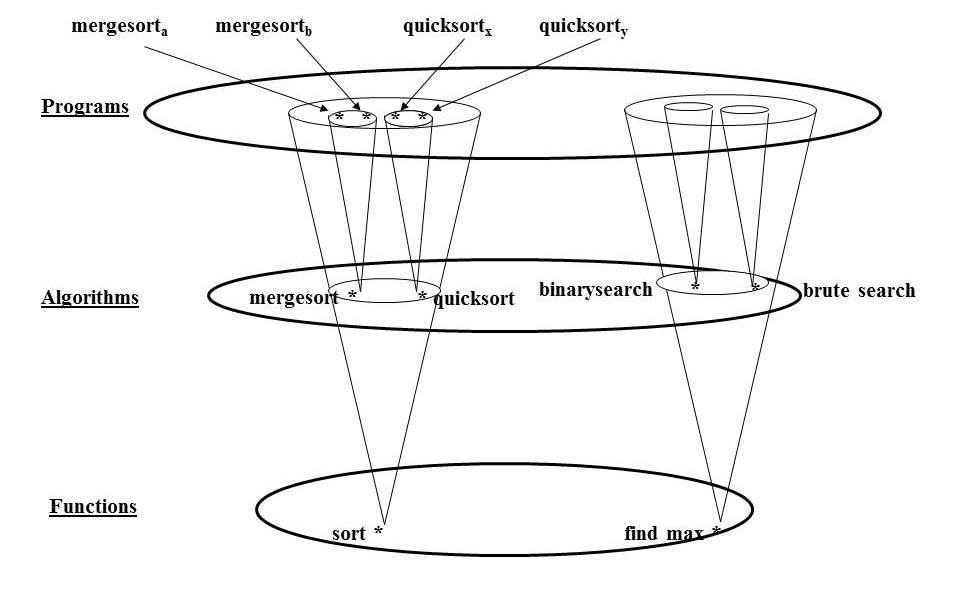}
  \caption{The Definition of an Algorithm} \label{pic:defofanalg}
\end{figure}

In terms of categories, this idea describes two full (symmetric monoidal) functors of (symmetric monoidal) categories. 
\be \xymatrix{
\Program\ar@{->>}[rr]&&
\Algorithm\ar@{->>}[rr]&&
\CompFunc 
}\ee
All the categories have sequences of types as objects. ($\Program$ is equivalent to $\Turing$ , $\RegMachine$ and $\CircuitFam$. They all fail to be true categories for the same reason as discussed in Technical Point \ref{tech:TuringNotCat}.) The left functor takes every program to the algorithm it implements. This functor is the identity on objects and full on morphism. The right functor takes every algorithm to the computable function that it describes. It too is the identity on objects and full on morphisms. 

There is something subjective about the question ``when are two programs considered ``essentially the same?'' Each answer will give us different categories of algorithms. See \cite{mygalois} for more about this. 

The top level, programs, is the domain of programmers. The bottom level, computable functions, is really what theoretical computer scientists study. And the middle level is the core of  computer science. The categories $\Program$, $\Turing$, $\RegMachine$ and $\CircuitFam$ are syntactical in the sense that you can write them down exactly. In contrast, $\CompFunc$ and the other categories in the center of The Big Picture are semantical. They are the  meaning of the relationship between inputs and outputs. Algorithms are somewhat in-between syntax and semantics. They are ``an idea'' of the method of going from input to output.  

Defining an object as an equivalence class of more concrete objects is not unusual. (i) Some philosophers follow Gottlob Frege in defining natural numbers as equivalence classes of finite sets that are bijective to each other. In detail, take the set of finite sets and put an equivalence relation: two sets are deemed equivalent if there exists a bijection between them. Every equivalence class corresponds to a natural number.  (As category theorists, we say that the natural numbers is the skeletal category of finite sets.) The number 3 is ``implemented'' by all the sets with three elements. (ii) Mathematicians describe a rational number as an equivalence class of pairs of integers. In detail, the pair $(x,y)$ is equivalent to $(x',y')$ if and only if $xy'=yx'$. The fraction $\f{1}{3}$ is ``implemented'' by the pairs $(1,3)$, $(10,30)$, $(-30,-90)$, $(534, 1602)$, etc. (iii) Physicists do not study physical phenomena. Rather they study collections of phenomena. That is, they look at all phenomena and declare two phenomena to be equivalent if there is some type of symmetry between them. Two experiments occur in different places, or are oriented differently, or occur at different times are considered the same if their outcome is the same. Laws of nature describe collections of physical phenomena, not individual phenomena. See \cite{mywhymath} for more about this and the relationship between collections of phenomena and mathematics.  

\section{Further Reading}
We have only scratched the surface. Theoretical computer science is an immense subject. We can only point the way for the reader to learn more. For a popular, non-technical introduction to much of this see Chapters 5 and 6 of my \cite{myolr} and David Harel's \cite{harel}.  
\begin{itemize}
\item Models of Computation: Every book in theoretical computer science has their favorite model of computation. Many use Turing machines for historical reasons. Sipser \cite{sipser}, Lewis and Papadimitriou \cite{lewis}, and Boolos, Burgess and Jeffrey \cite{boolos} all use Turing machines.
 There is more about the development of the Turing machine idea in Andrew Hodges' excellent biography of Alan Turing \cite{hodges}.   Register machines can be found in \cite{cutland, davis,rogers}. 
\item Computability Theory: There are many excellent books in this area, e.g., \cite{sipser, cutland, davis}. Much can be learned about oracle computation and the whole hierarchy of unsolvable problems in \cite{soare} and in \cite{rogers}. 
\item Complexity Theory: Some textbooks are \cite{lewis, papa, barak} and Chapter 7-10 of \cite{sipser}. There is much about NP-Compete problems in \cite{garey}.
\item Kolmogorov Complexity Theory: The main textbook in this field is \cite{li}. Christian S. Calude's book \cite{caludeinformation} is wonderful. There is also a short, beautiful introduction to the whole field in Section 6.4 of \cite{sipser}. 
One of the founders of this field is Gregory J. Chaitin. All his books and papers are interesting and worth studying.
\end{itemize}

None of the above sources mention any category theory. Our presentation is novel in that these topics have not been presented before in a uniform way using categories. 

The idea of defining an algorithm as equivalence class of programs comes from my paper \cite{mydefalg}. There is a followup to the paper which deals with many different equivalence classes \cite{mygalois}. The first paper was criticized by Andreas Blass, Nachum Dershowitz,  and Yuri Gurevich \cite{blass}. My definition of an algorithm is used in the second edition of Manin's logic book \cite{manin} and by several others since.

This mini-course uses category theory as a bookkeeping tool to store and compare all the various parts of theoretical computer science. There is, however,  a branch of research that uses category theory in a deeper way. They describe properties of categories that would be able to deal with computations. Perhaps the first work in this direction was done by one of the founders of category theory, Sammy Eilenberg. Towards the end of his career, in 1970,  he and Calvin C. Elgot published a small book titled ``Recursiveness'' \cite{EilenbergS:rec}. In 1974 and 1976 he published a giant two-volume work on formal language theory titled {\it Automata, languages, and machines} \cite{sammya, sammyb}. Giuseppe Longo and  Eugenio Moggi also had several papers in this direction \cite{longo84, longo84a, longo90}. 
In 1987, Alex Heller (who was my thesis advisor and a dear friend) and Robert DiPaola (who was a teacher of mine and a close friend) published a paper ``Dominical categories: recursion theory without elements''\cite{hellerD}.  This work was followed by papers of Heller \cite{Heller} and Florian Lengyel \cite{florian}. There are various types of similar categories with names like ``P-Categories'', ``Restriction categories'',`` Recursion categories'', and  ``Turing Categories''.  See Robin Cockett and Pieter Hofstra's paper \cite{cockett} for a clear history of the development of these ideas.
Dusko Pavlovic develops the notion of a computation in a monoidal category in series of papers that start here \cite{dusko}. 
There is also a development of such ideas for complexity theory in ``Otto's thesis'' \cite{otto} and in paper by Ximo Diaz-Boils \cite{ximo}.

Yuri Manin's paper \cite{ManinCC} and his subsequent book \cite{manin} put all computations into one category called a ``computational universe''. This is similar to what is done in this mini-course.  He was also able to incorporate quantum computing into his categories. 

Another connection between theoretical computer science and category theory is implementing categorical structures on computers. The first place to look for this is in Rydeheard and Burstal's  {\it Computational Category Theory} \cite{rydeheard}

It is worth mentioning yet another connection between theoretical computer science and category theory. I wrote a paper \cite{mycomplexity} which shows that there are constructions in category theory that can mimic the workings of a Turing machine. Since limits and colimits are infinitary operations, it is possible for categories to ``solve'' the Halting problem. (But this solution cannot be implemented on a finite computer.)

\newpage

\markboth{Bibliography}{Bibliography }

\bibliography{TCSfWCTbib}
\bibliographystyle{plain}

\end{document}